\documentclass{sig-alternate}
\usepackage{graphicx}
\usepackage{hyperref}
\usepackage{verbatim}

\makeatletter
\newif\if@restonecol
\makeatother

\usepackage[ruled,noend,linesnumbered,lined]{algorithm2e}

\newtheorem{theorem}{Theorem}
\newtheorem{example}{Example}
\newtheorem{definition}{Definition}

\newtheorem{lemma}[theorem]{Lemma}
\newtheorem{proposition}[theorem]{Proposition}
\newtheorem{corollary}[theorem]{Corollary}

\usepackage{stmaryrd}

\DeclareMathOperator*{\argmax}{arg\,max}

\newcommand{\sem}[1]{\llbracket#1\rrbracket}

\newcommand{\sched}{\text{\textsc{Schedulable}}}
\newcommand{\nosched}{\text{\textsc{UnSchedulable}}}

\newcommand{\discr}[1]{\mathit{discr}(#1)}
\newcommand{\facet}[2]{#1|_{#2}}
\newcommand{\rows}[1]{\mathit{rows}(#1)}
\newcommand{\ball}[2]{B_{#1}(#2)}
\newcommand{\External}{\mathsf{External}}
\newcommand{\Hh}{\mathcal{H}}
\newcommand{\Ii}{\mathcal{I}}

\newcommand{\Rr}{\mathcal{R}}

\newcommand{\set}[1]{\left\{ #1 \right\}}

\newcommand{\seq}[1]{\langle #1 \rangle}
\newcommand{\Rplus}{{\mathbb R}_{\geq 0}}
\newcommand{\Rpos}{{\mathbb R}_{> 0}}
\newcommand{\Nat}{\mathbb N}
\newcommand{\Real}{\mathbb R}

\newcommand{\FRUNS}{\text{\it FRuns}}
\newcommand{\RUNS}{\text{\it Runs}}
\newcommand{\RUN}{\text{\it Run}}

\newcommand{\SCH}{\mathcal{W}_\mathrm{Safe}}

\newcommand{\Ext}{\text{\it Ext}}
\newcommand{\Vtx}[1]{\mathsf{vert}(#1)}

\newcommand{\point}[1]{{\overline{#1}}}

\newcommand{\pc}{\point{c}}
\newcommand{\px}{\point{x}}
\newcommand{\pv}{\point{v}}
\newcommand{\py}{\point{y}}

\newcommand{\vv}{\vec{v}}
\newcommand{\vu}{\vec{u}}
\newcommand{\vw}{\vec{w}}
\newcommand{\vb}{\vec{b}}

\newcommand{\vr}{\vec{r}}
\newcommand{\vp}{\vec{p}}
\newcommand{\vn}{\vec{n}}

\newcommand{\norm}[1]{\|#1\|}
\newcommand{\dotprod}{\cdot}
\newcommand{\vzero}{\vec{0}}
\DeclareMathOperator{\interior}{int}

\newcommand{\cmms}{\textsf{CMS}}
\newcommand{\bmms}{\textsf{BMS}}
\newcommand{\mmms}{\textsf{MMS}}

\usepackage{color}

\usepackage{tikz}
\tikzstyle{nloc}=[draw,circle,minimum size=4em,inner sep=0em]
\tikzstyle{mloc}=[draw,circle,minimum size=2em,inner sep=0em]
\tikzstyle{trans}=[-latex, rounded corners]

\def\techreport{}

\begin{document}
\title{Safe Schedulability of Bounded-Rate Multi-Mode Systems}  
\numberofauthors{4}
\author{
  \alignauthor
  Rajeev Alur\\
  \affaddr{University of Pennsylvania, Philadelphia, USA}\\
  \email{alur@seas.upenn.edu}
  \alignauthor
  Vojt\v{e}ch Forejt\\
  \affaddr{Dept. of Computer Science, University of Oxford, UK}\\
  \email{vojfor@cs.ox.ac.uk}
  \alignauthor
  Salar Moarref\\
  \affaddr{University of Pennsylvania, Philadelphia, USA}\\
  \email{moarref@seas.upenn.edu}
  \and  %
  \alignauthor
  Ashutosh Trivedi\\
  \affaddr{Indian Institute of Technology, Bombay}\\
  \email{trivedi@cse.iitb.ac.in}
}
\date{\today}
\maketitle

\begin{abstract}
  Bounded-rate multi-mode systems (\bmms{}) are hybrid systems that can switch freely
  among a finite set of modes, and whose dynamics is specified by a finite
  number of real-valued variables with mode-dependent rates that can vary
  within given bounded sets.
  The schedulability problem for \bmms{} is
  defined as an infinite-round game between two players---the scheduler and the
  environment---where in each round the scheduler proposes a time and a mode
  while the environment chooses an allowable rate for that mode, and the state
  of the system changes linearly in the direction of the rate vector. 
  The goal of the scheduler is to keep the state of the system within a
  pre-specified safe set using a non-Zeno schedule, while the goal of the
  environment is the opposite. 
  Green scheduling under uncertainty is a paradigmatic example of \bmms{}
  where a winning strategy of the scheduler corresponds to
  a robust energy-optimal policy.
  We present an algorithm to decide whether the scheduler has a winning
  strategy from an arbitrary starting state, and give an algorithm to compute
  such a winning strategy, if it exists. 
  We show that the schedulability problem for \bmms{} is co-NP complete in
  general, but for two variables it is in PTIME.  
  We also study the discrete schedulability problem where the
  environment has only finitely many choices of rate vectors in each
  mode and the scheduler can make decisions only at multiples of a given clock
  period, and show it to be EXPTIME-complete.
\end{abstract}




\section{Introduction}
\label{Sec:introduction}
There is a growing trend towards multi-mode compositional design
frameworks~\cite{LP12,PLS10,LOTM12} for the synthesis of cyber-physical systems
where the desired system is built by composing various modes, subsystems, or motion
primitives---with well-understood performance characteristics---so as
to satisfy certain higher level control objectives.     
A notable example of such an approach is \emph{green scheduling} proposed by
Nghiem et al.~\cite{NBMJ11,NBPM11} where the goal is to 
compose different modes of heating, ventilation, and air-conditioning (HVAC)
installations in a building so as to keep the temperature surrounding each
installation in a given comfort zone while keeping the peak energy consumption
under a given budget.  
Under the assumption that the state of the system grows linearly in each mode,
Nghiem et al. gave a polynomial algorithm to decide the green schedulability
problem.
Alur, Trivedi, and Wojtczak~\cite{ATW12} studied general constant-rate
multi-mode systems and showed, among others, that the result of Nghiem et
al. holds for arbitrary multi-mode systems with constant rate dynamics as long
as the scheduler can switch freely among the finite set of modes.  

In this paper we present \emph{bounded-rate multi-mode systems} that generalize
constant-rate multi-mode systems by allowing non-constant mode-dependent rates
that are given as bounded polytopes. 
Our motivations to study bounded-rate multi-mode schedulability are twofold. 
First, it allows one to model a conservative approximation of green
schedulability problem in presence of more complex inter-mode dynamics.
Second motivation is theoretical and it stems from the desire to characterize
decidable problems in context of design and analysis of cyber-physical systems.
In particular, we view a bounded-rate multi-mode system as a two-player
extension of constant-rate multi-mode system, and show the decidability of 
schedulability game for such systems.

Before discussing bounded-rate multi-mode system (\bmms{}) in any further
detail, let us review the definition, relevant results, and limitations of
constant-rate multi-mode system (\cmms{}).  
A \cmms{} is specified as a finite set of variables whose dynamics in a finite set
of modes is given as mode-dependent constant rate vector.
The \emph{schedulability problem} for a \cmms{} and a bounded convex safety set
of states is to decide whether there exists an infinite sequence (schedule) of
modes and time durations such that choosing modes for corresponding time
durations in that sequence keeps the system within the safety set forever.   
Moreover such schedule is also required to be physically implementable, i.e. the
sum of time durations must diverge (the standard non-Zeno
requirement~\cite{HLMR05}). 
Alur et al.~\cite{ATW12} showed that, for the starting states in the interior of
the safety set, the necessary and sufficient condition for safe
schedulability is the existence of an assignment of dwell times to modes such
that the sum of rate vectors of various modes scaled by corresponding dwell time
is zero. 
Intuitively, if it is possible using the modes to loop back to the starting state,
i.e. to go to some state other than the starting state and then to return to the
starting state, then the same schedule can be scaled appropriately and repeated
forever to form a \emph{periodic schedule} that keeps the system inside the
interior of any convex safety set while ensuring time divergence.
On the other hand, if no such assignment exists then Farkas' lemma implies
the existence of a vector such that choosing any mode the system makes a
positive progress in the direction of that vector, and hence for any non-Zeno
schedule the system will leave any bounded safety set in a finite amount of
time. 
Also, due to constant-rate dynamics such condition can be modeled as 
a linear program feasibility problem, yielding a polynomial-time
algorithm. 
\begin{example}
\label{ex11}
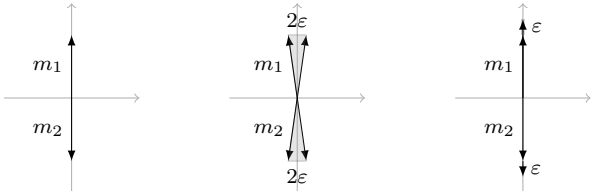
\begin{figure}
  \centering
  {\small
  \begin{tikzpicture}[node distance=2cm,scale=0.6, y=0.7cm]
    \tikzstyle{lines}=[draw=black!30,rounded corners]
    \tikzstyle{vectors}=[-latex, rounded corners]
    \tikzstyle{rvectors}=[draw=black!30,very thick, rounded corners]

    \draw[lines,->] (-1.5,0)--(1.5,0);
    \draw[lines,->] (0, -3)--(0,3);
    \draw[vectors] (0, 0) --node[left]{$m_{1}$} (0, 2) node[left]{$$};
    \draw[vectors] (0, 0) --node[left]{$m_{2}$} (0, -2)node[right]{$$};

    \draw[lines,->] (3.5,0)--(6.5,0);
    \draw[lines,->] (5,-3)--(5,3);

    \draw[vectors] (5, 0) --node[left]{$$} (5.2, 2) node[left]{$$};
    \draw[vectors] (5, 0) --node[left]{$m_1$} (4.8, 2) node[left]{$$};
    \draw[fill=black!50,opacity=0.2] (5.2,2) -- (4.8,2) -- (5, 0);
    \draw (5, 2) node[above]{$2\varepsilon$};

    \draw[vectors] (5, 0) --node[left]{$$} (5.2, -2) node[left]{$$};
    \draw[vectors] (5, 0) --node[left]{$m_2$} (4.8, -2) node[left]{$$};
    \draw[fill=black!50,opacity=0.2] (5.2,-2) -- (4.8,-2) -- (5, 0);
    \draw (5, -2) node[below]{$2\varepsilon$};

    \draw[lines,->] (8.5,0)--(11.5,0);
    \draw[lines,->] (10, -3)--(10,3);

    \draw[vectors] (10, 0) --node[left]{$m_1$} (10, 2) node[left]{$$};
    \draw[rvectors] (10, 2) --node[right]{$\varepsilon$} (10, 2.5);
    \draw[vectors] (10, 0) -- (10, 2.5);
    \draw[vectors] (10, 0) --node[left]{$m_2$} (10, -2) node[left]{$$};
    \draw[vectors] (10, -2) -- node[right]{$\varepsilon$} (10, -2.5);

  \end{tikzpicture}
}
  \vspace{-1.5em}
  \caption{Multi-mode systems with uncertain rates}
  \vspace{-1em}
  \label{geom2}
\end{figure}
Consider the 2-dimensional \cmms{} shown in Figure~\ref{geom2}~(left) with two modes
$m_1$ and $m_2$ with rates of the variables as $\vr_1 = (0, 1)$ in mode $m_1$ and
$\vr_2 = (0, -1)$ in mode $m_2$.
It is easy to see that the system is schedulable for any starting state 
$(x_0, y_0)$ in the interior of any bounded convex set $S$ as 
$\vr_1 + \vr_2 = (0, 0)$.
The safe schedule consists of the periodic schedule $(m_1, t), (m_2, t)$ for a
carefully selected $t \in \Rpos$ such that $(x_0, y_0) + \vr_1 t$ stays
inside $S$.
\end{example}
However, the schedules constructed in this manner are not robust as an arbitrarily small
change in the rate can make the schedule unsafe as shown in the following example.  
\begin{example}
  Consider a multi-mode system where some environment related
  fluctuations~\cite{HLMR05} cause the rate vectors in modes $m_1$ and $m_2$ to 
  differ from those in Example~\ref{ex11} by an arbitrarily small $\varepsilon >
  0$ as shown in Figure~\ref{geom2}~(middle).
  Here, $m_1$ can have rate-vectors from 
  $\set{(0{+}\delta, 1) \::\: -\varepsilon {\leq} \delta {\leq} \varepsilon}$, while
  rate-vectors of $m_2$ are from 
  $\set{(0{+}\delta, -1) \::\: -\varepsilon {\leq} \delta {\leq} \varepsilon}$.
  First we show that the periodic schedule $(m_1, t), (m_2, t)$ proposed in 
  Example~\ref{ex11} is not safe for any $t$.
  Consider the case when the rate vector in modes $m_1$ and $m_2$ are fixed to
  $(\varepsilon, 1)$ and $(\varepsilon, -1)$. 
  Starting from the state $(x_0, y_0)$ and following the periodic schedule
  $(m_1, t), (m_2, t)$ for $k$ steps the state of the system will be $(x_0 +k t
  \varepsilon, y_0)$ after $k$ steps.
  Hence it is easy to see that for any bounded safety set the state of the
  system will leave the safety set after finitely many steps.
  In fact, for this choice of rate vectors no non-Zeno safe schedule exists at
  all, since by choosing any mode for a positive time the system makes a
  positive progress along the $X$ axis.  
\end{example}
We formalize modeling of such multi-mode system under uncertainty as
bounded-rate multi-mode systems (\bmms{}).
\bmms{}s can also approximate~\cite{Hen96} the effect of more complex
non-linear, and even time-varying, mode dynamics over a bounded safety set. 
Formally, a \bmms{} is specified as a finite set of variables whose dynamics in
a finite set of modes is given as a mode-dependent bounded convex polytopes of rate vectors.
We present the schedulability problem on \bmms{} as an infinite-round
zero-sum game between two players, \emph{the 
  scheduler} and \emph{the environment}; at each round scheduler chooses a
mode and a time duration, the environment chooses a   
rate vector
from the allowable set
of rates for that mode, and the state of the system is evolved accordingly.  
The recipe for selecting their choices, or \emph{moves}, is formalized in
the form of a strategy that is a function of the history of the game so far to a
move of the player. 
A strategy is called \emph{positional} if it a function only of the current
state. 
We say that the scheduler wins the schedulability game, or has a
winning strategy, from a given starting state if there is a scheduler strategy
such that irrespective of the strategies of the environment the state of the
system stays within the safety set and time does not converge to any real number. 
Similarly, we say that the environment has a winning strategy if she has a
strategy such that for any strategy of the scheduler the system leaves the
safety set in a finite amount of time, or the time converges to some real number. 
One of the central results of this paper is that the schedulability games on
\bmms{} are \emph{determined}, i.e. for each starting state exactly
one of the player has a winning strategy.
Note that the determinacy of these games could be proved using more general results on
determinacy, e.g. ~\cite{MS98}, however our proof is direct and shows the
existence of positional winning strategies.

We distinguish between two kind of strategies of scheduler--the \emph{static
  strategies}, where scheduler can not observe the decisions of the environment,
and the \emph{dynamic strategies}, where scheduler can observe the decisions of the
environment so far before choosing a mode and a time.
Static strategies correspond precisely to schedules, and we often use these two 
terms interchangeably. 
A key challenge in the schedulability analysis of \bmms{} is that static
strategies are not sufficient as is clear from the following example.
\begin{example}
  \label{ex22}
  Consider the \bmms{} of Figure~\ref{geom2}~(right) where the
  rates in mode $m_1$ and $m_2$ lie in $\set{(0, 1+\delta) \::\: 0 \leq \delta
    \leq \varepsilon}$ and $\set{(0, -(1+\delta)) \::\: 0 \leq \delta
    \leq \varepsilon}$, respectively.
  We hint that there is no static winning strategy of scheduler in this \bmms{} 
  (the formal conditions on where the static winning strategy exists will be
  analyzed later in the paper). 
  Let us assume, for example,
  that $\sigma = (m_1, t_1), (m_2, t_2),
  \ldots$ is a static non-Zeno winning strategy of the scheduler. 
  Moreover consider two strategies $\pi$ and $\pi'$ of the environment that
  differ only in mode $m_1$ where they propose rates $(1,0)$ and
  $(1+\varepsilon,0)$ respectively.
  Let $\varrho$ and $\varrho'$ be the sequences of system states and player's
  choices---what we subsequently refer to as runs---as the game progresses from
  a starting state $(x_0, y_0)$ where the environment uses strategy $\pi$ and
  $\pi'$, respectively, against scheduler's strategy $\sigma$. 
  Let $T_1(i)$ and $T_2(i)$ be the time spent in mode $m_1$ and
  $m_2$, resp., till the $i$-th round in runs $\varrho$ and $\varrho'$, while
  $T_1$ and $T_2$ be total time spent in mode $m_1$ and $m_2$, resp.
  The state of the system in the runs $\varrho$ and $\varrho'$ after $i$ rounds will
  be $(x_0, y_0 + T_1(i) - T_2(i))$ and $(x_0, y_0 + T_1(i) - T_2(i) +
  T_1(i) \varepsilon)$.
  Hence the distance $T_1(i) \varepsilon$ between states reached after
  $i$-rounds in runs $\varrho$ and $\varrho'$ tends to $T_1 \varepsilon$ as $i$ tends to
  $\infty$. 
  It is easy to see that if $\sigma$ is a winning strategy then $T_1 {=} \infty$;
  since if $T_1 {<}\infty$ and $T_2 {=} \infty$ then the system will move in the
  direction of rates of mode $m_2$, while if both $T_1$ and $T_2$ are finite
  then the strategy is not non-Zeno.
  Hence system will eventually leave any bounded safety set, contradicting 
  our assumption on $\sigma$ being a winning strategy. 
\end{example}

The techniques used for schedulability analysis and schedule construction for
\cmms{} cannot be generalized to \bmms{} since in a \bmms{}, the scheduler may
not have a strategy to loop back to the starting state.
In fact, in general scheduler does not have a strategy to revisit any 
state as is clear from Figure~\ref{geom2}~(right)---here the
environment can always choose a rate vector in both mode $m_1$ and $m_2$ to
avoid any previously visited state.
However, from our results on \bmms{} it follows that if the scheduler has a
winning strategy then he has a strategy to restrict the future states of the
system to a ball of arbitrary diameter centered around the starting state.

In order to solve schedulability game for \bmms{} we exploit the following
observation: the scheduler has a winning strategy, from all the starting states
in the interior of the safety set $S$, if and only if there is a polytope $P
\subseteq S$, such that for every vertex $\pv$ of $P$ there is a mode $m(\pv)$
and time $t(\pv)$ such that choosing mode $m(\pv)$ for time $t(\pv)$ from
the vertex $\pv$, the line $\pv + \vr t(\pv)$ stays within polytope $P$ for all 
allowable rates $\vr$ of $m(\pv)$.  
In other words, for any vertex of $P$ there is a mode and a time duration
such that if the system evolves with any rate vector of that mode for such
amount of time, the system stays in $P$.   
For a \bmms{} $\Hh$ we call such a polytope \emph{$\Hh$-closed}.
The \emph{$\Hh$-closed} polytope is similar to controlled invariant set in control theory literature 
(see ~\cite{Blanchini19991747} for a comprehensive review). 
We show how such a polytope can be constructed for a BMS based on its characteristics. 
 We also analyze the complexity of such a construction.
The existence of an $\Hh$-closed polytope immediately provides a non-Zeno safe
dynamic strategy for the scheduler for any starting state in $P$: find
the convex coefficient $(\lambda_1, \lambda_2, \ldots, \lambda_k)$ of the
current state $\px$ with respect to the finite set of vertices $(\px_1, \px_2,
\ldots, \px_k)$ of $P$ and choose the mode $m(\px_i)$ for time $t(\px_i)
\lambda_i$ that maximizes $t(\px_i) \lambda_i$.
Then, for some choice $\vr$ of the environment for $m(\px_i)$ the system will
progress to $\px' {=} \px {+} t(\px_i) \lambda_i \vr$.   
One can repeat this dynamic strategy from the next state $\px'$ as the current
state. 
We prove that such strategy is both non-Zeno and safe.

An extreme-rate \cmms{} of a \bmms{} $\Hh$ is obtained by preserving the set of
modes, and for each mode assigning a rate which is a vertex of the
available rate-set of that mode. 
The main result of the paper is that an $\Hh$-closed polytope exists for a
\bmms{} $\Hh$ iff all \emph{extreme-rate} \cmms{}s of $\Hh$ are
schedulable.  
The ``only if'' direction of the above characterization is immediate as if some
extreme-rate \cmms{} is not schedulable then the environment can fix those rate
vectors and win the schedulability game in the \bmms{}.
We show the ``if'' direction by explicitly constructing the $\Hh$-closed
polytope. 

\begin{example}
  \label{mainex}
  Consider the \bmms{} $\Hh$ from Figure~\ref{geom2}~(right) with $\varepsilon=0.5$.  
  The safety set is given as a shaded area in Figure \ref{strat-tree}~(left) and
  $\bar{x}_0=(-1,-0.5)$ is the initial state. 
  Observe that all extreme-rate combinations are schedulable and hence we show a
  winning strategy.
  An $\Hh$-closed polytope for this \bmms{}  is the line-segment between the
  points $(0,2.5)$ and $(0,-2.5)$ (we explain the construction of such polytope
  in Section~\ref{sec:schedulability}). 
  After translating this line-segment to $x_0$ and scaling it to fit inside the
  safety set, we will get the line-segment connecting $\bar{x}_1=(-1,1)$ to  
  $\bar{x}_2=(-1,-2)$, as shown in Figure \ref{strat-tree}~(left). 
  At vertices $\bar{x}_1$ and $\bar{x}_2$ modes $m_2$ and $m_1$, respectively, can be used
  for $1$ time unit.  
  A winning strategy of scheduler is to keep the system's state along the line
  segment. 
  Our strategy observes the current state $\px$ and finds the mode to
  choose by computing convex coefficient $\lambda {\in} [0,1]$ s.t.
  $\px {=}\lambda \px_1 {+}(1{-}\lambda) \px_2$. 
  For instance, at state $\bar{x}_0{=}\frac{1}{2}\bar{x}_1{+}\frac{1}{2}\bar{x}_2$ the
  scheduler can choose any of the modes for $\frac{1}{2}$ time units. 
  Assume that it chooses $m_1$. 
  Based on environment's choice the state of system after $\frac{1}{2}$ time
  units will be in the set $\set{-1,0.5+\delta \::\: 0 \leq \delta \leq 0.5}$. 
  The scheduler observes this new state after $\frac{1}{2}$ time-unit, and chooses
  mode and time accordingly. 
  For example, if the environment chooses $(0,1.25)$ and so the next state is
  $\bar{x}=(-1,0.75)=\frac{1}{12}\bar{x}_1{+}\frac{11}{12}\bar{x}_2$,
  scheduler can choose mode $m_2$ for $\frac{11}{12}$ time units. 
  In Figure~\ref{strat-tree}~(right) we show  first two rounds of the game.  
  Since, for any point on our line segment scheduler can choose a mode for at least
  $0.5$ time unit and stay on the line segment, such strategy is both safe and non-Zeno.
  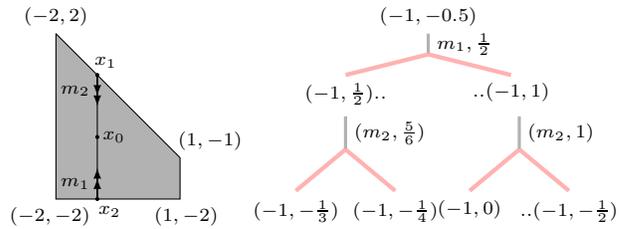
\begin{figure}
    \centering
    {\scriptsize
  \begin{tikzpicture}[node distance=2cm,scale=0.55]
    \tikzstyle{lines}=[draw=black,rounded corners]
    \tikzstyle{vectors}=[-latex, rounded corners]
    \tikzstyle{rvectors}=[draw=black!30,very thick]
    \tikzstyle{gvectors}=[draw=red!30,ultra thick]
    \tikzstyle{rrvectors}=[draw=blue!30,ultra thick]

    \draw[lines,sharp corners,fill=black!30] 
    (-2, 2) node[above] {$(-2, 2)$}-- 
    (1, -1) node[above] {$\phantom{xxxxx}(1, -1)$} -- 
    (1, -2) node[below] {$\phantom{x}(1, -2) $} --
    (-2, -2) node[below] {$(-2, -2)\phantom{x}$} --(-2,2);

    \draw[lines] (-1,-2)--(-1,1);

    \draw[vectors] (-1, 1) --node[left]{$$} (-1, 0.5) node[left]{$$};
    \draw[vectors] (-1, 1) --node[left]{$m_2$} (-1, 0.25);

    \draw[vectors] (-1, -2) --node[left]{$$} (-1, -1.5) node[left]{$$};
    \draw[vectors] (-1, -2) --node[left]{$m_1$} (-1, -1.25);
    
    \node[draw,circle,inner sep=0.4pt,fill] at (-1,1) {};
    \node[draw,circle,inner sep=0.4pt,fill] at (-1,-2) {};
    \node[draw,circle,inner sep=0.4pt,fill] at (-1,-0.5){};
    \draw (-0.6, -0.5) node {$x_0$};
    \draw (-0.8, 1.3) node {$x_1$};
    \draw (-1, -2.3) node {$\phantom{xx}x_2$};

    \draw (7, 2) node[above] {$(-1, -0.5)$};
    
    \draw[rvectors] (7, 2) -- node[right]{$ m_1, \frac{1}{2}$} (7, 1.5);
    \draw[gvectors] (7, 1.5) -- node[right, below]{$$}(5, 1);
    \draw[gvectors] (7, 1.5) -- node[left, below]{$$} (9, 1);

    \draw (5, 1) node[below] {$(-1, \frac{1}{2})..$};
    \draw (9, 1) node[below] {$..(-1, 1)$};

    \draw[rvectors] (5, 0) -- node[right]{$(m_2, \frac{5}{6})$}(5, -0.8);
    \draw[rvectors] (9.2, 0) -- node[right]{$(m_2, 1)$} (9.2, -0.8);

    \draw[gvectors] (5, -0.8) -- (3.8, -1.8) node[below]{$(-1, -\frac{1}{3})$} ;
    \draw[gvectors] (5, -0.8) -- (6.2, -1.8) node[below]{$(-1, -\frac{1}{4})$};

    \draw[gvectors] (9.2, -0.8) -- (8, -1.8) node[below]{$(-1, 0)$} ;
    \draw[gvectors] (9.2, -0.8) -- (10.4, -1.8) node[below]{$..(-1, -\frac{1}{2})$};

  \end{tikzpicture}
}
  \vspace{-2em}
    \caption{$\Hh$-closed polytope and dynamic strategy}
    \label{strat-tree}
  \vspace{-1em}
  \end{figure}
\end{example}

We also extend the above result to decide the winner starting from arbitrary
states, i.e. including those states that lie on the boundary of the safety
set. 
Here we show that the existence of a safe scheduler implies the existence of a
safe scheduler which only allows to move from lower-dimensional faces to
higher-dimensional ones and not the other way around; this allows us to use an
algorithm which traverses the face lattice of the safety set and analyses each
face one by one. 
We also prove co-NP completeness of the schedulability problem,
showing the hardness by giving a
reduction from 3-SAT to the non-schedulability problem.
On a positive note, we show that if the number of variables is two,
then the schedulability game can be decided in polynomial
time. This is because in such a case we can prove that there is only
polynomially many candidates for falsifiers we need to consider, and hence we can
check each of them one by one.
Finally, we study a discrete version of schedulability games where scheduler
can choose time delays only at multiples of a given clock period, while the
environment can choose rate vectors from a finite set.
We show that discrete schedulability games on $\bmms{}$ are EXPTIME-complete,
and that the maximal clock period for which scheduler has a winning strategy can
be computed in exponential time. 
If the system is a \cmms{}, we get a PSPACE algorithm, improving the
result of~\cite{ATW12} where only an approximation of the maximal clock period
for \cmms{} was studied.

We refer to~\cite{NBPM11,NBMJ11} and~\cite{ATW12} for a review of related work on
\cmms{} and green scheduling.  
Heymann et al.~\cite{HLMR05} considered scheduling problem on \bmms{} where
rate-vectors are given as upper and lower rate matrices and the safety set as
the entire non-negative orthant. 
They showed that the scheduler wins if he wins in the \cmms{} of the lower rate
matrix, and wins only if he wins in the \cmms{} of the upper rate matrix.  
We study more general \bmms{} and safety sets, and characterize 
necessary and sufficient condition for schedulability.
To complete the picture, we remark that games on hybrid
automata~\cite{HK99,HHM99}, that corresponds to \bmms{} with
local invariants and guards, have undecidable schedulability problem.

\section{Problem Definition}
\label{sec:prelims}

{\bf Points and Vectors.} 
Let $\Real$ be the set of real numbers.
We represent the states in our system as points in $\Real^n$ that is equipped
with the standard \emph{Euclidean norm} $\norm{\cdot}$.
We denote points in this state space by $\px, \py$,  vectors by $\vr, \vv$, and 
the $i$-th coordinate of point $\px$ and vector $\vr$ by $\px(i)$ and $\vr(i)$,
respectively. 
We write $\vzero$ for a vector with all its coordinates equal to $0$; 
its dimension is often clear from the context.
The distance $\norm{\px, \py}$ between points $\px$ and $\py$ is defined as
$\norm{\px - \py}$. 
For two vectors ${\vv_1, \vv_2 \in \Real^n}$, we write $\vv_1 \dotprod \vv_2$
to denote their dot product defined as $\sum_{i=1}^n \vv_1(i)\cdot\vv_2(i)$.

{\bf Boundedness and Interior.} 
We denote a {\em closed ball} of radius $d \in \Rplus$ centered at $\px$ as
$\ball{d}{\px} {=} \set{\py {\in} \Real^n \::\: \norm{\px,\py} \leq d}$.
We say that a set $S \subseteq \Real^n$ is {\em bounded} if there exists
$d \in \Rplus$ such that for all $\px, \py \in S$ we have
$\norm{\px,\py} \leq d$.
The {\em interior} of a set $S$, $\interior(S)$, is the set of all points
$\px \in S$ for which there exists $d > 0$ s.t. $\ball{d}{\px} \subseteq S$.

{\bf Convexity.} A point $\px$ is a \emph{convex
  combination} of a finite set of points $X = \set{\px_1, \px_2, \ldots, \px_k}$ if
there are $\lambda_1, \lambda_2, \ldots, \lambda_k \in [0, 1]$ such that
$\sum_{i=1}^{k} \lambda_i = 1$ and $\px = \sum_{i=1}^k \lambda_i \cdot \px_i$.
The \emph{convex hull} of $X$ is then the set of all points
that are convex combinations of points in $X$.
We say that $S \subseteq \Real^n$ is {\em convex} iff for all
$\px, \py \in S$ and all $\lambda \in [0,1]$ we have
$\lambda \px + (1-\lambda) \py \in S$ and moreover,
$S$ is a {\em convex polytope} if it is bounded and there exists $k \in \Nat$, a matrix $A$ of 
size $k \times n$ and a vector $\vb \in \Real^k$ such that $\px \in S$ iff
$A\px \leq \vb$. 
We write $\rows{M}$ for the number of rows in a matrix $M$, here $\rows{A}=k$. 

A point $\px$ is a \emph{vertex} of a convex polytope $P$ if it is
not a convex combination of two distinct (other than $\px$) points in $P$. 
For a convex polytope $P$ we write $\Vtx{P}$ for the finite set of points that
correspond to the vertices of $P$.  
Each point in $P$ can be written as a convex combination of the points in
$\Vtx{P}$, or in other words, $P$ is the \emph{convex hull} of $\Vtx{P}$.
From standard properties of polytopes, it follows that for every convex polytope
$P$ and every vertex $\pc$ of $P$, there exists a vector $\vv$ such that
$\vec{v}\cdot \pc=d$ and $\vec{v}\cdot \px>d$  for all $\px\in P \setminus
\set{\pc}$ for some $d$.
We call such a vector $\vv$ a \emph{supporting hyperplane} of the polytope
$P$ at $\pc$. 

\subsection{Multi-Mode Systems}
A multi-mode system is a hybrid system equipped with finitely many \emph{modes}
and finitely many real-valued \emph{variables}. 
A configuration is described by values of the variables, which change, as the
time elapses, at the rates determined by the modes being used.
The choice of rates is nondeterministic, which introduces
a notion of adversarial behavior.
Formally,
\vspace{-1em} 
\begin{definition}[Multi-Mode Systems]
  \label{def:BMMS}
  A multi-mode system is a tuple $\Hh = (M, n, \Rr)$ where:
    $M$ is the finite nonempty set of \emph{modes}, 
    $n$ is the number of continuous variables, and 
    $\Rr : M \to 2^{\Real^n}$ is the rate-set function that, for each mode $m \in 
    M$, gives a set of vectors.
\end{definition}

We often write  $\vr \in m$ for $\vr\in \Rr(m)$ when $\Rr$ is clear
from the context. 
A finite \emph{run} of a multi-mode system $\Hh$ is a finite sequence of states,
timed moves and rate vector choices  
$\varrho = \seq{\px_0, (m_1, t_1), \vr_1, \px_1,  \ldots, (m_k, t_k), \vr_k,
  \px_k}$ 
s.t. for all $1 \leq i \leq k$ we have $\vr_i \in \Rr(m_i)$ and   
$\px_i {=} \px_{i-1} + t_i \cdot \vr_i$.
For such a run $\varrho$ we say that $\px_0$ is the \emph{starting state}, while
$\px_k$ is its \emph{last state}.
An infinite run is defined in a similar manner. 
We write $\RUNS$ and $\FRUNS$ for the set of infinite and finite runs of $\Hh$,
while $\RUNS(\px)$ and $\FRUNS(\px)$  for the set of infinite and finite runs
starting from $\px$.  

An infinite run $\seq{\px_0, (m_1, t_1), \vr_1, \px_1, (m_2, t_2),
  \vr_2,  \ldots}$ is \emph{Zeno} if $\sum_{i=1}^{\infty} t_i < \infty$.
Given a set $S \subseteq \Real^n$ of safe states, we say that a  run
$\seq{\px_0, (m_1, t_1), \vr_1, \px_1, (m_2, t_2), \vr_2,  \ldots}$ is
$S$-safe if for all $i \geq 0$ we have that $\px_i \in S$ and $\px_i + t \cdot
\vr_{i+1} \in S$ for all $t \in [0, t_{i+1}]$, assuming $t_0 = 0$.
Notice that if $S$ is a convex set and
$\px_i \in S$ for all $i\ge 0$, then for all $i \geq 0$ and for all $t \in
[0, t_{i+1}]$ we have that $\px_i + t \cdot \vr_{i+1} \in S$.
The concept of $S$-safety for finite runs is defined in a similar manner.
Sometimes we simply call a run safe when the safety set and the starting state
is clear from the context. 

We formally give the semantics of a multi-mode system $\Hh$ as a
turn-based two-player game between the players, \emph{scheduler} and
\emph{environment}, who choose their moves to construct a run of the system.   
The system starts in a given starting state $\px_0 \in \Real^n$ and at each turn
scheduler chooses a timed move, a pair $(m, t) \in M \times \Real_{>0}$ consisting
of a mode and a time duration, and the environment chooses a rate vector 
$\vr \in \Rr(m)$ and as a result the system changes its state from $\px_0$ to
the state $\px_1 = \px_0 + t \cdot \vr$ in $t$ time units following the linear
trajectory according to the rate vector $\vr$.
From the next state $\px_1$ the scheduler again chooses a timed move and the
environment an allowable rate vector, and the game continues forever in this
fashion.
The focus of this paper is on safe-\emph{schedulability game}, where
the goal of the scheduler is to keep the states of the system within a given
safety set $S$, while ensuring that the time diverges (non-Zenoness requirement).   
The goal of the environment is the opposite, i.e. to visit a state out of the
safety set or make the time converge to some finite number.

Given a bounded and convex safety set $S$, we define \emph{(safe) schedulability
objective} $\SCH^S$ as the set of $S$-safe and non-Zeno runs of $\Hh$. 
In a schedulability game the winning objective of the scheduler is to make sure
that the constructed run of a system belongs to $\SCH^S$, while the goal of the
environment is the opposite. 
The choice selection mechanism of the players is typically defined as
strategies. 
A \emph{strategy} $\sigma$ of scheduler is function  $\sigma{:} \FRUNS {\to} M
{\times} \Rplus$ that gives a timed move for every history of the game.
A strategy $\pi$ of the environment is a function $\pi: \FRUNS \times (M \times
\Rplus) \to \Real^n$ that chooses an allowable rate for a given history of
the game and choice %
of the scheduler.
We say that a strategy is \emph{positional} if it suggests the same action for
all runs with common last state.
We write $\Sigma$ and $\Pi$ for the set of strategies of the scheduler and the
environment, respectively.

Given a starting state $\px_0$ and a strategy pair $(\sigma, \pi) \in \Sigma
\times \Pi$ we define the unique run $\RUN(\px_0, \sigma, \pi)$ 
starting from $\px_0$ as 
\[
\RUN(\px_0, \sigma, \pi)=\seq{\px_0, (m_1, t_1), \vr_1, \px_1, (m_2, t_2),
  \vr_2,  \ldots}
\]
where for all $i{\ge} 1$, $(m_i, t_i) = \sigma(\seq{\px_0, (m_1, t_1), \vr_1, \px_1, \ldots,\px_{i-1}})$ and 
$\vr_i=\pi(\seq{\px_0, (m_1, t_1), \vr_1, \px_1, \ldots,\px_{i-1},m_i, t_i})$ and
$x_i= x_{i-1} + t_i\cdot \vr_i$.
The scheduler wins the game if there is $\sigma \in \Sigma$ such that for all $\pi \in \Pi$
we get $\RUN(\px_0, \sigma, \pi) \in \SCH^S$.   
Such a strategy $\sigma$ is \emph{winning}.  
Similarly, the environment wins the game if there is $\pi \in \Pi$ such that
for all $\sigma \in \Sigma$ we have $\RUN(\px_0, \sigma, \pi) \not \in \SCH^S$.
Again, $\pi$ is called \emph{winning} in this case.
If a winning strategy for scheduler exists, we say that $\Hh$ is schedulable for
$S$ and $\px_0$ (or simply {\em schedulable} if $S$ and $\px_0$ are clear from
the context). 
The following is the main algorithmic problem studied in this paper.

\begin{definition}[Schedulability]
  \label{def:schedulability}
  Given a multi-mode system $\Hh$, a safety set $S$, and a starting
  state $\px_0 \in S$, the (safe) \emph{schedulability problem} is to decide whether
  there exists a winning strategy of the scheduler.
\end{definition}

\subsection{Bounded-Rate Multi-Mode Systems} 
To algorithmically decide schedulability problem, we need to restrict the range
of $\Rr$ and the domain of safety set $S$ in a schedulability game on a
multi-mode system. 
The most general model that we consider is the bounded-rate multi-mode systems
(\bmms{}) that are multi-mode systems $(M, n, \Rr)$ such that $\Rr(m)$ is a
convex polytope for every $m \in M$. 
We also assume that the safety set $S$ is specified as a convex polytope.
In our proofs we often refer to another variant of multi-mode systems in which
there are only a fixed number of different rates in each mode (i.e. $\Rr(m)$ is
finite for all $m\in M$).   
We call such a multi-mode system \emph{multi-rate multi-mode systems}
(\mmms{}). 
Finally, a special form of \mmms{} are  \emph{constant-rate multi-mode systems}
(\cmms{})~\cite{ATW12} in which $\Rr(m)$ is a singleton for all $m\in M$.
We sometimes use $\Rr(m)$ to refer to the unique element of the set $\Rr(m)$ in
a \cmms{}.
The concepts for the schedulability games for \bmms{} and \mmms{} are already
defined for multi-mode systems. 
Similar concepts also hold for \cmms{} but note that the environment has no real
choice in this case. 
For this reason, we can refer to a schedulability game on \cmms{} as a
one-player game.  

The prime~\cite{ATW12} practical motivation for studying \cmms{} was to
generalize results on green scheduling problem by Nghiem et al.~\cite{NBPM11}. 
We argue that \bmms{} are a suitable abstraction to study green scheduling problem
when various rates of temperature change are either uncertain or follow a complex
and time-varying dynamics, as shown in the following example.

\begin{example}[Green Scheduling]
  \label{example1}
  Consider a building with two rooms $A$ and $B$.  
  HVAC units in each zone can be in
  one of the two modes $0$ (OFF) and $1$ (ON). 
  We write the mode of the combined system as $m_{i,j}$ to represent 
  the fact that rooms $A$ and $B$ are in mode $i \in \set{0, 1}$ and $j \in
  \set{0,1}$, respectively. 
  The rate of temperature change and the energy usage for each room is given
  below. 
  \begin{center}
      \begin{tabular}{|p{4.5cm}|l|l|l|}
        \hline
        Zones& ON & OFF \\
        \hline
        A (temp. change rate/ usage) & -2/2& 2/1\\
      \hline
      B (temp. change/ usage) & -2/2& 2/1\\
      \hline
    \end{tabular}
  \end{center}
  Following~\cite{ATW12} we assume that the energy cost is equal to energy
  usage if peak energy usage at any given point in time is less than or equal to
  $3$ units, otherwise energy cost is $10$ times of that standard rate.
  It follows that to minimize energy cost the peak usage, if possible, must not
  be higher than $3$ units at any given time.
  We can model the system as a \cmms{} with modes $m_{0, 0}$, $m_{0, 1}$, and
  $m_{1, 0}$, because these are the only ones that have peak usage at most $3$.  
  The variables of the \cmms{} are the temperature of the rooms, while the
  safety set is the constraint that temperature of both zones should be between
  $65^o F$ to $75^o F$.
  The existence of a winning strategy in \cmms{} implies the existence of a
  switching schedule with energy peak demand less than or equal to $4$ units.
  In Figure~\ref{fig:example1}.(a) we show a graphical representation of such
  \cmms{} with three modes $m_{0,0}, m_{0,1}$ and $m_{1,0}$ and two variables
  (corresponding to the two axes).  
  The rate of the variables in mode $m_{0,0}$ is $(2, 2)$, in mode $m_{0,1}$ is
  $(2,-2)$, and in mode $m_{1,0}$ is $(-2, 2)$.
  \begin{figure}[t]
    \centering
    {\small
  \begin{tikzpicture}[node distance=2cm,scale=0.4]
    \tikzstyle{lines}=[draw=black!30,rounded corners]
    \tikzstyle{vectors}=[-latex, rounded corners]
    \tikzstyle{rvectors}=[-latex,very thick, rounded corners]

    \draw[lines] (-3.2,0)--(3.2,0);
    \draw[lines] (0, 3.2)--(0,-3);
    \draw[vectors] (0, 0) --node[right]{$m_{0,0}$} (2, 2) node[left]{$$};
    \draw[vectors] (0, 0) --node[right]{$m_{0,1}$} (2, -2)node[right]{$$};
    \draw[vectors] (0, 0) --node[left]{$m_{1,0}$} (-2, 2)node[right]{$$};

    \draw (0, -4) node {$(a)$ Constant-Rate};

    \draw[lines] (4.2,0)--(10.2,0);
    \draw[lines] (7, 3.2)--(7,-3);
    
    \draw[vectors] (7, 0) --node[right]{$$} (8.5, 1.5) node[left]{$$};
    \draw[vectors] (7, 0) --node[right]{$$} (8.5, 2) node[left]{$$};
    \draw[vectors] (7, 0) --node[right]{$$} (9, 2) node[left]{$$};
    \draw[vectors] (7, 0) --node[right]{$m_{0,0}$} (9, 1.5) node[left]{$$};
    \draw[fill=black!50,opacity=0.2] (8.5, 1.5) -- (8.5, 2) -- (9, 2) -- (9,
    1.5) -- (8.5, 1.5); 
    
    \draw[vectors] (7, 0) --node[right]{$$} (8.5, -2.5)node[right]{$$};
    \draw[vectors] (7, 0) --node[right]{$$} (8.5, -2)node[right]{$$};
    \draw[vectors] (7, 0) --node[right]{$$} (9, -2)node[right]{$$};
    \draw[vectors] (7, 0) --node[right]{$m_{0,1}$} (9, -2.5)node[right]{$$};
    \draw[fill=black!50,opacity=0.2] (8.5, -2.5) -- (8.5, -2) -- (9, -2) -- (9,
    -2.5) -- (8.5, -2.5); 
    
    \draw[vectors] (7, 0) --node[left]{$$} (4.5, 1.5)node[right]{$$};
    \draw[vectors] (7, 0) --node[left]{$m_{1,0}$} (4.5, 2)node[right]{$$};
    \draw[vectors] (7, 0) --node[left]{$$} (5, 2) node[right]{$$};
    \draw[vectors] (7, 0) --node[left]{$$} (5, 1.5) node[right]{$$};
    \draw[fill=black!50,opacity=0.2] (4.5, 1.5) -- (4.5, 2) -- (5, 2) -- (5,
    1.5) -- (4.5, 1.5); 
    
    \draw (7, -4) node {$(b)$ Bounded-Rate};

   \draw[lines] (11.2,0)--(17.2,0);
    \draw[lines] (14, 3.2)--(14,-3);

    \draw[vectors] (14, 0) --node[right]{$$} (15.5, 1.5) node[left]{$$};
    \draw[vectors] (14, 0) --node[right]{$$} (15.5, 2) node[left]{$$};
    \draw[vectors] (14, 0) --node[right]{$$} (16, 2) node[left]{$$};
    \draw[vectors] (14, 0) --node[right]{$m_{0,0}$} (16, 1.5) node[left]{$$};
    
    \draw[vectors] (14, 0) --node[right]{$$} (15.5, -2.5)node[right]{$$};
    \draw[vectors] (14, 0) --node[right]{$$} (15.5, -2)node[right]{$$};
    \draw[vectors] (14, 0) --node[right]{$$} (16, -2)node[right]{$$};
    \draw[vectors] (14, 0) --node[right]{$m_{0,1}$} (16, -2.5)node[right]{$$};
    
    \draw[vectors] (14, 0) --node[left]{$$} (11.5, 1.5)node[right]{$$};
    \draw[vectors] (14, 0) --node[left]{$m_{1,0}$} (11.5, 2)node[right]{$$};
    \draw[vectors] (14, 0) --node[left]{$$} (12, 2)node[right]{$$};
    \draw[vectors] (14, 0) --node[left]{$$} (12, 1.5)node[right]{$$};

    \draw (14, -4) node {$(c)$ Multi-Rate };

  \end{tikzpicture}
}
    \vspace{-2em}
    \caption{ \label{fig:example1} Restricted Multi-mode Systems}
    \vspace{-1em}
  \end{figure}
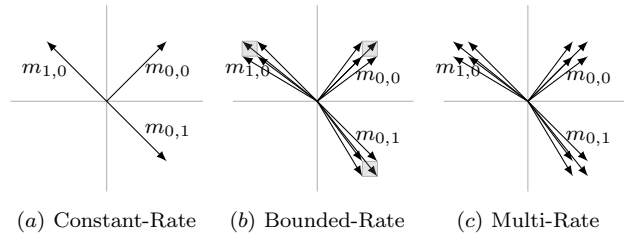

  Now assume that the rate of temperature change in a mode is not constant
  and can vary within a given margin $\varepsilon >0$.
  Schedulability problem for such system can best be modeled as a \bmms{} as
  shown in Figure~\ref{fig:example1}.(b) where the polytope of possible rate
  vectors is shown as a shaded region.
  In Figure~\ref{fig:example1}.(c) we show a \mmms{} where variables can only
  change with the extreme rates of the \bmms{} in  Figure~\ref{fig:example1}.(b).
\end{example}

We say that a \cmms{} $H = (M, n, R)$ is an instance of a multi-mode system 
$\Hh = (M, n, \Rr)$ if for every $m \in M$ we have that $R(m) \in \Rr(m)$.
For example, the \cmms{} shown in Figure~\ref{fig:example1}.(a) is an instance
of \bmms{} in Figure~\ref{fig:example1}.(b).
We denote the set of instances of a multi-mode system $\Hh$ by $\sem{\Hh}$. 
Notice that for a \bmms{} $\Hh$ the set $\sem{\Hh}$ of its instances is
uncountably infinite, while for a \mmms{} $\Hh$ the set $\sem{\Hh}$ is finite
whose size is exponential in the size of $\Hh$. 
We say that a \mmms{} $(M, n, \Rr')$ is the \emph{extreme-rate} \mmms{} of a
\bmms{} $(M, n, \Rr)$ if $\Rr'(m) = \Vtx{\Rr(m)}$. 
The \mmms{} in Figure~\ref{fig:example1}.(c) is the extreme-rate \mmms{} for the
\bmms{} in Figure~\ref{fig:example1}.(b) 
We write $\Ext(\Hh)$ for the extreme-rate \mmms{} of the \bmms{} $\Hh$.

Notice that for every starting state and winning objective at most one player
can have a winning strategy.
We say that a game is not \emph{determined} if no player has a winning strategy
for some starting state. 
In the next section we give an algorithm to decide the winner in a
schedulability game for an arbitrary starting state. 
Since for every starting state we can decide the winner, it gives a direct proof
of determinacy of schedulability games on \bmms{}. 
Moreover, it follows from our results that whenever a player has a winning
strategy, he has a positional such strategy. 
These two results together yield the first key results of this paper. 
\begin{theorem}[Determinacy]
  \label{thm:determinacy}
  Schedulability games on \bmms{} with convex safety polytopes are
  positionally determined. 
\end{theorem}

In Section~\ref{sec:complexity} we analyze the complexity of deciding the winner
in a schedulability game. 
Using a reduction from SAT problem to non-schedulability for a \mmms{}, we prove
the following main contribution of the paper. 
\begin{theorem}\label{thm:coNP}
  Schedulability problems for \bmms{} and $\mmms{}$ are co-NP complete.
\end{theorem}
On a positive note, we also show that schedulability games can be solved in
polynomial time for \bmms{} and \mmms{} with two variables.

\section{Solving Schedulability Games}
\label{sec:schedulability}
In this section we discuss the decidability of the schedulability problem for
\bmms{}.  
We first present a solution for the case when the starting state is in the
interior of a safety set, and generalize it to arbitrary starting states 
in Section~\ref{sec:sched-general}.

\subsection{Starting State in the Interior of Safety Set}\label{sec:interior}
Alur et al.~\cite{ATW12} presented a polynomial-time algorithm to decide if the
scheduler has a winning strategy in a schedulability game on a \cmms{} for an
arbitrary starting state. 
In particular, for starting states in the interior of the safety set, they
characterized a necessary and sufficient condition. 
\begin{theorem} [\cite{ATW12}]
  \label{thmcms}
  The scheduler has a winning strategy in a \cmms{} $(M, n, R)$, with convex
  safety set $S$ and starting state $\px_0$ in the interior of $S$, iff
  there is  $\vec{t} \in \Rplus^{|M|}$ satisfying:
  \begin{equation}\label{eq:cmms}
    \sum_{i=1}^{|M|} R(i)(j) \cdot \vec{t}(i) = 0 \text{ for $1\leq j \leq n$ and }
    \sum_{i=1}^{|M|} \vec{t}(i) = 1.
  \end{equation}
\end{theorem}

We call a \cmms{} \emph{safe} if it satisfies (\ref{eq:cmms}) and we call $H$
unsafe otherwise. 
The intuition behind Theorem~\ref{thmcms} is that the scheduler has a winning
strategy if and only if it is possible to return to the starting state in
strictly positive time units. 
From the results of~\cite{ATW12} it also follows that whenever a winning
strategy exists, there is a strategy which does not look at a history or even
the current state, but only uses a bounded counter of size $\ell \le |M|-1$ and
after after a history of length $k$ makes a decision only based on the number
$k$ modulo $\ell$. 
Such strategies are called \emph{periodic}.

It is natural to ask whether the approach of~\cite{ATW12} can be generalized 
to \bmms{}. 
Unfortunately, Example~\ref{ex22} shows that in a \bmms{} although a winning
strategy may exist, it may not be possible to return to the initial state, or
indeed visit any state twice. 
Another natural question to ask is whether a suitable generalization of periodic
strategies suffice for \bmms{}. 
\emph{Static} strategies are \bmms{} analog of periodic strategies that behave
in the same manner irrespective of the choices of the environment, i.e. for a
static strategy $\sigma$ we have that 
$\sigma(\rho) = \sigma(\rho')$ for all runs $\rho = \seq{\px_0, (m_1, 
  t_1), \vr_1, \px_1,  \ldots, (m_k, t_k), \vr_k, \px_k}$ and 
$\rho' = \seq{\px_0, (m_1, t_1), \vr'_{1}, \px_1',  \ldots, (m_k, t_k), \vr'_{k}),
  \px_k'}$.
Static strategies are often desirable in the settings where scheduler can
not observe the state of the system.
However, as we show in
\ifthenelse{\isundefined{\techreport}}{%
\cite{techreport}}{%
Appendix~\ref{app:static}}%
, except for the degenerate
cases when the \bmms{} contains a subset of modes which induce a safe \cmms{},
scheduler can never win a game on \bmms{} using static strategies. 
We saw an example of this phenomenon in the Introductory section as
Figure~\ref{geom2}.(c). 

This negative observations imply that to solve the schedulability games for \bmms{}
one needs to take a different approach. 
In the rest of this section, we define the notion of $\Hh$-closed polytope and
show that if such a polytope exists, then for any convex set $S$ we can construct a
winning \emph{dynamic} strategy which takes its decisions only based on the last
state. 
We also extend the notion of safety of a \cmms{} to \bmms{}. 
We say that a \bmms{} $\Hh$ is \emph{safe} if all instances of its extreme-rate \mmms{}
$\Ext(\Hh)$ are safe i.e. all $H \in \sem{\Ext(\Hh)}$ satisfy~(\ref{eq:cmms}).
Finally, we connect (Lemmas~\ref{lem:sched-pushing} and~\ref{lem:extreme}) the
existence of $\Hh$-closed polytope with the safety of the \bmms{}. 

{\bf Dynamic Scheduling Algorithm.}
For a \bmms{} $\Hh$ we call a convex polytope $P$ {\em $\Hh$-closed}, if for
every vertex of $P$ there exists a mode $m$ such that all the rate
vectors of $m$ keep the system in $P$, i.e.
for all $\pc \in \Vtx{P}$ there exists $m \in M$ and $\tau \in \Rpos$ such that for
all $\vr \in \Rr(m)$ we have that $\pc + \vr\cdot t \in P$ for all $t \in [0,
\tau]$.
An example of a $\Hh$-closed polytope is given in the
Example~\ref{mainex}.

\begin{algorithm}[t]
  \KwIn{BMMS $\Hh$, starting state $\px_0$}
  \KwOut{non-Terminating Scheduling Algorithm}
  $\gamma := $ the shortest distance of $\px_0$ from borders of $S$\;
  $P := $ $\Hh$-closed polytope s.t. $P\subseteq\ball{\gamma}{\px_0}$ and
  $\px_0\in P$\nllabel{line:get-polytope}\;  
  \ForEach{$\point c\in \Vtx{P}$}{
    \ForEach{mode $m \in M$}{
      \ForEach {extreme rate vector $\vr \in m$}{
        $t_{\vr}=\max \{t \::\: \point c+\vr\cdot t \in P\}$\; 
      }
      $\delta_m=\min_{\vr\in m} t_{\vr}$\;
    }
    $m_* = \argmax_{m\in M} \delta_m$; \quad
    $\Delta_{\point c} = \delta_{m_*}$; \quad
    $m_{\point{c}}=m_*$\;
  } 
  \While{true}{
    Store current state as $\px$\; 
    Find $(\lambda_{\point c} \geq 0)_{\point c\in \Vtx{P}}$ where
    $\px=\sum_{\pc \in \Vtx{P}} \lambda_{\point c} \cdot \point c$\; 
    Find $\pc_*=\argmax_{\pc \in \Vtx{P}} \lambda_{\pc}\cdot \Delta_{\pc}$\;
    Schedule mode $m_{\point c_*}$ for
    $\lambda_{\pc_*}\cdot\Delta_{\pc_*}$\nllabel{line:time}\; 
  }
  \caption{Dynamic scheduling algorithm}
  \label{SchedulingAlg}
\end{algorithm}

Assume that for any $\gamma>0$ and $\px_0$ we are able to compute a $\Hh$-closed
polytope which is fully contained in $\ball{\gamma}{\px_0}$ and contains $\px_0$.
If this is the case, we can use Algorithm~\ref{SchedulingAlg} to compute a
dynamic scheduling strategy. 
The idea of the algorithm is to build a $\Hh$-closed polytope which contains the
initial state and is fully contained within $S$, and then construct the strategy
based on the modes safe at the vertices of the polytope.
The correctness of the algorithm is established by the following proposition.
\begin{proposition}\label{prop:dynsched}
  If there exists an $\Hh$-closed polytope and it can be effectively computed
  then Algorithm~\ref{SchedulingAlg} implements a winning dynamic strategy for
  the scheduler.     
\end{proposition}
\begin{proof}
  Assume that there exists an $\Hh$-closed polytope and we have an algorithm to
  compute it. 
  Observe that the strategy is non-Zeno, because
  $\lambda_{\pc_*}\cdot\Delta_{\pc_*}$ on line~\ref{line:time} is bounded from
  below by $\frac{1}{|\Vtx{P}|}\cdot \min_{\pc \in \Vtx{P}} \Delta_{\pc}$ for
  any point of $P$, and $\Delta_{\pc}$ are positive by their construction and the
  definition of the $\Hh$-closed polytope. 
  Next, we need to show that under the computed strategy we never leave the
  convex polytope $P$. 
  For a state $\px$ which is of the form $\sum_{\pc \in \Vtx{P}} \lambda_{\pc}\cdot \pc$,
  the successor state will be
  $\px'=(\sum_{\pc \in \Vtx{P}} \lambda_\pc\cdot \pc)+ \lambda_{\pc_*}\cdot\Delta_{\pc_*}\cdot\vr$
  where $\vr$ is the rate picked by the environment. 
  We can rewrite $\px'$ as 
  $(\sum_{\pc \in \Vtx{P}\setminus \{\pc_*\}} \lambda_\pc\cdot \pc)+
  \lambda_{\pc_*}\cdot (\pc_* + \vr\cdot \Delta_{\pc_*})$.
  Since $\pc_* + \vr\cdot \Delta_{\pc_*}\in P$, we get that $\px'$ is a
  convex combination of points in $P$ and hence lies in $P$.
\end{proof}

{\bf Constructing $\Hh$-Closed Polytope.}
\begin{figure*}[t]
  \centering
  {\small
  \begin{tikzpicture}[node distance=2cm,scale=0.45]
    \tikzstyle{lines}=[draw=black!30,rounded corners]
    \tikzstyle{vectors}=[-latex, rounded corners]
    \tikzstyle{rvectors}=[-latex,very thick, rounded corners]

    \draw[lines] (-13.4,0)--(-5.2,0);
    \draw[lines] (-9.3, 2.5)--(-9.3,-2.5);
    \draw[vectors] (-9.3, 0) --node[left]{$m_{1}$} (-9.3, 2) node[left]{$$};

    \draw[fill=black!50,opacity=0.2] (-9.3,0) -- (-11.3, 0) -- (-11.3, -1)  -- (-9.3, 0);

    \draw[vectors] (-9.3, 0) --node[below]{$$} (-11.3, 0)node[right]{$$};
    \draw[vectors] (-9.3, 0) --node[above,left]{$m_2$} (-11.3, -1)node[right]{$$};

    \draw[vectors] (-9.3, 0) --node[right]{$$} (-7.3, 0)node[right]{$$};
    \draw[vectors] (-9.3, 0) --node[above,right]{$m_3$} (-7.3, -1)node[right]{$$};

    \draw[fill=black!50,opacity=0.2] (-9.3,0) -- (-7.3, 0) -- (-7.3, -1)  -- (-9.3, 0);

    \draw (-9.3, -3.5) node {$(a)$};

    \draw[lines] (-4.2,0)--(4.2,0);
    \draw[lines] (0, 2.5)--(0,-2.5);
    \draw[vectors] (0, 0) --node[left]{$$} (0, 2) node[left]{$$};
    \draw[vectors] (0, 0) --node[left]{$$} (-2, 0)node[right]{$$};
    \draw[vectors] (0, 0) --node[right]{$$} (-2, -1)node[right]{$$};
    \draw[vectors] (0, 0) --node[right]{$$} (2, 0)node[right]{$$};
    \draw[vectors] (0, 0) --node[left]{$$} (2, -1)node[right]{$$};

     \node[draw,circle,inner sep=0.25pt,fill] at (-2,2) {};
     \node[draw,circle,inner sep=0.25pt,fill] at (0,2) {};
     \node[draw,circle,inner sep=0.25pt,fill] at (2,2) {};
     \node[draw,circle,inner sep=0.25pt,fill] at (-4,1) {};
     \node[draw,circle,inner sep=0.25pt,fill] at (-2,1) {};
     \node[draw,circle,inner sep=0.25pt,fill] at (2,1) {};
     \node[draw,circle,inner sep=0.25pt,fill] at (4,1) {};
     \node[draw,circle,inner sep=0.25pt,fill] at (-2,0) {};
     \node[draw,circle,inner sep=0.25pt,fill] at (2,0) {};
     \node[draw,circle,inner sep=0.25pt,fill] at (-4,-1) {};
     \node[draw,circle,inner sep=0.25pt,fill] at (-2,-1) {};
     \node[draw,circle,inner sep=0.25pt,fill] at (2,-1) {};
     \node[draw,circle,inner sep=0.25pt,fill] at (4,-1) {};
     \node[draw,circle,inner sep=0.25pt,fill] at (-2,-2) {};
     \node[draw,circle,inner sep=0.25pt,fill] at (0,-2) {};
     \node[draw,circle,inner sep=0.25pt,fill] at (2,-2) {};

    \draw (0, -3.5) node {$(b)$};

      \draw[lines] (5,0)--(13.4,0);
      \draw[lines] (9.2, 2.5)--(9.2,-2.5);
      
      \node[draw,circle,inner sep=0.25pt,fill] at (7.2,2) {};
     \node[draw,circle,inner sep=0.25pt,fill] at (9.2,2) {};
     \node[draw,circle,inner sep=0.25pt,fill] at (11.2,2) {};
     \node[draw,circle,inner sep=0.25pt,fill] at (5.2,1) {};
     \node[draw,circle,inner sep=0.25pt,fill] at (7.2,1) {};
     \node[draw,circle,inner sep=0.25pt,fill] at (11.2,1) {};
     \node[draw,circle,inner sep=0.25pt,fill] at (13.2,1) {};
     \node[draw,circle,inner sep=0.25pt,fill] at (7.2,0) {};
     \node[draw,circle,inner sep=0.25pt,fill] at (11.2,0) {};
     \node[draw,circle,inner sep=0.25pt,fill] at (5.2,-1) {};
     \node[draw,circle,inner sep=0.25pt,fill] at (7.2,-1) {};
     \node[draw,circle,inner sep=0.25pt,fill] at (11.2,-1) {};
     \node[draw,circle,inner sep=0.25pt,fill] at (13.2,-1) {};
     \node[draw,circle,inner sep=0.25pt,fill] at (7.2,-2) {};
     \node[draw,circle,inner sep=0.25pt,fill] at (9.2,-2) {};
     \node[draw,circle,inner sep=0.25pt,fill] at (11.2,-2) {};
      
     \draw[fill=black!50,opacity=0.2] (5.2,1) -- (7.2,2) -- (11.2,2)  -- (13.2,1) -- (13.2,-1) --
    (11.2,-2) -- (7.2, -2) -- (5.2,-1) -- (5.2,1);

    \draw (9.2, -3.5) node {$(c)$};

    \draw[lines] (14.4,0)--(22.8,0);
    \draw[lines] (18.6, 2.5)--(18.6,-2.5);

    \node[draw,circle,inner sep=0.25pt,fill] at (14.6,1) {};
     \node[draw,circle,inner sep=0.25pt,fill] at (16.6,2) {};
     \node[draw,circle,inner sep=0.25pt,fill] at (20.6,2) {};
     \node[draw,circle,inner sep=0.25pt,fill] at (22.6,1) {};
     \node[draw,circle,inner sep=0.25pt,fill] at (22.6,-1) {};
     \node[draw,circle,inner sep=0.25pt,fill] at (20.6,-2) {};
     \node[draw,circle,inner sep=0.25pt,fill] at (16.6, -2) {};
     \node[draw,circle,inner sep=0.25pt,fill] at (14.6,-1) {};
    
    \draw[fill=black!50,opacity=0.2] (14.6,1) -- (16.6,2) -- (20.6,2)  -- (22.6,1) -- (22.6,-1) --
    (20.6,-2) -- (16.6, -2) -- (14.6,-1) -- (14.6,1);

    \draw[vectors] (14.6,1) --node[left]{$$} (16.6, 1)node[right]{$$};
    \draw[vectors] (14.6,1) --node[left]{$$} (16.6, 0)node[right]{$$};
    \draw[vectors] (16.6,2) --node[left]{$$} (18.6, 2)node[right]{$$};
    \draw[vectors] (16.6,2) --node[left]{$$} (18.6, 1)node[right]{$$};
    \draw[vectors] (20.6,2) --node[left]{$$} (18.6, 2)node[right]{$$};
    \draw[vectors] (20.6,2) --node[left]{$$} (18.6, 1)node[right]{$$};
    \draw[vectors] (22.6,1) --node[left]{$$} (20.6, 1)node[right]{$$};
    \draw[vectors] (22.6,1) --node[left]{$$} (20.6, 0)node[right]{$$};
    \draw[vectors] (22.6,-1) --node[left]{$$} (22.6, 0)node[right]{$$};
    \draw[vectors] (22.6,-1) --node[left]{$$} (22.6, 0)node[right]{$$};
    \draw[vectors] (22.6,-1) --node[left]{$$} (22.6, 0)node[right]{$$};
    \draw[vectors] (20.6,-2) --node[left]{$$} (20.6, -1)node[right]{$$};
    \draw[vectors] (16.6, -2) --node[left]{$$} (16.6, -1)node[right]{$$};
    \draw[vectors] (14.6, -1) --node[left]{$$} (14.6, 0)node[right]{$$};
    
    \draw (18.6, -3.5) node {$(d)$};

  \end{tikzpicture}
}
  \caption{Constructing closed convex polytope}
  \label{geom3}
\end{figure*}
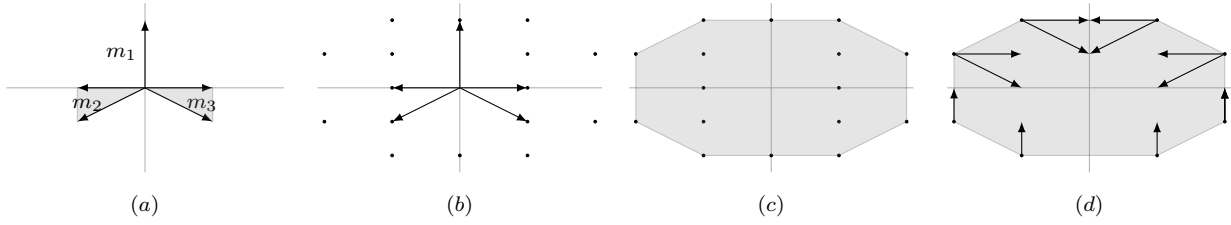
We will next show how to implement line~\ref{line:get-polytope} of
Algorithm~\ref{SchedulingAlg}. 
We give necessary and sufficient conditions for existence of $\Hh$-closed
polytopes in the following two lemmas. 
The first lemma shows that an $\Hh$-closed polytope exists if and only if for
any hyperplane (given by its normal vector $\vv$) there exists a mode $m$ such
that all its rates stay at one side of the hyperplane. 
\begin{lemma}\label{lem:sched-pushing}
  For a \bmms{} $\Hh$, a state $\px_0$ and $\gamma > 0$, there is a $\Hh$-closed
  polytope $P\subseteq \ball{\gamma}{\px_0}$ with $\px_0 \in P$ if and only if
  for every $\vv$ there is a mode $m$ such that $\vv\cdot \vr \ge 0$  for all $\vr
  \in m$.  
\end{lemma}
\begin{proof}
  Let us fix a \bmms{} $\Hh = (M, n, \Rr)$.    
  The proof is in two parts. 
  For $\Rightarrow$, assume that the system is schedulable but there exists a
  vector $\vv$ such that for all modes $m\in M$ there is a rate $\vr_m\in m$
  where  ${\vv\cdot \vr_m < 0}$.  
  It implies that if the adversary fixes the rates $\vr_m$ whenever the scheduler
  chooses $m$, then the system moves in the direction of vector
  $-\vv$ (i.e. for all $d$ a state $\px$ will be reached such that ${\vv\cdot \px < d}$),
  and hence for any bounded safety set and non-Zeno strategy
  system will leave the safety set. 
  This contradicts with existence of $\Hh$-closed polytope implying winning
  scheduler strategy. 
  
  To prove the other direction, let  $R {=} \set{\vr_1,\ldots,\vr_N}$ be
  the set of rates occurring in modes of the extreme-rate \mmms{} of $\Hh$,
  i.e. 
  $R = \set{\Rr'(m) : (M,n,\Rr')\in\sem{\Ext(\Hh)},  m\in M} $. 
  We claim the following to be the $\Hh$-closed polytope:
  \begin{equation}
  P:= \{\px_0 + D\cdot \sum_{i=1}^N \vr_i \cdot p_i\mid p_i\in [0,1]\},\label{hpol}
  \end{equation}
  where $D=\gamma/\sum_{i=1}^N\norm{\vr_i}$. 
  Notice that $P$ is a convex polytope since it is a convex hull of points 
  ${\px_0 {+} D\cdot \sum_{i=1}^N \vr_i \cdot p_i}$ where $p_i\in \{0,1\}$.
  Also, due to our choice of $D$, $P\subseteq\ball{\gamma}{\px_0}$, and
  $\px_0\in P$. 
  For the sake of contradiction we assume that for every $\vv$ there is a mode
  $m$ such that all rates $\vr$ of $m$ satisfy $\vv\cdot \vr \ge 0$,
  but at least one corner
  $\pc$ of $P$ does not satisfy the defining condition of $\Hh$-closed polytope, i.e.
  for all modes $i$ there is a rate vector
  $\vr_i$ satisfying 
  \begin{eqnarray}
    \pc + t \cdot \vr_i \not \in  P \text{ for all $t > 0$} \label{eq0}
  \end{eqnarray}
  Let us fix such corner $\pc$.
  By the supporting hyperplane theorem there is a vector
  $\vv$ such that, for some $d$: 
  \begin{eqnarray}
    \vv\cdot \pc&=&d\label{eq1}\\
    \vv\cdot \px&>&d, \text{ for all $\px\in P \setminus
      \set{\pc}$}\label{eq2}
  \end{eqnarray}
  i.e. $\vv$ is supporting $P$ on $\pc$. 
  Let us fix some mode $m$ such that for all rates
  $\vr$ of $m$ we have $\vv\cdot \vr \ge 0$.
  Notice that this exists by the assumption. Let $\vr_i$ be a rate of $m$
  satisfying~(\ref{eq0}). 

  By the definition of $P$ the point $\pc$, a corner of $P$, is of the
  form $\px_0 + D\cdot \sum_{j=1}^N \vr_j \cdot p_j$ for some $p_j\in[0,1]$ where $1\le
  j\le N$ and $\vr_j \in R$. 
  We necessarily have $p_i=1$, because if $p_i=1-\delta$ for some
  $\delta>0$, then $\pc+D\cdot\varepsilon\cdot \vr_i \in P$ for any
  $\varepsilon\le \delta$ and that will contradict with (\ref{eq0}).
  Notice that for all $k \in [0, 1]$ the points 
  $\py_k = \px_0 + D\cdot\sum_{j=1}^N p^k_j\cdot \vr_j$, where $p^k_j = p_j$ if $j \neq i$ and
  $p^k_j = k$ otherwise, are all in $P$. 
  Also notice that point $\py_1 = \pc$ and for each $k \in [0, 1]$ we have
  that $\py_k = \py_0 + D\cdot k \cdot \vr_i$.
  In particular, $\pc = \py_1 = \py_0 + D\cdot \vr_i$.
  It follows that $\pc -D\cdot \vr_i = \py_0 \in P$.
  W.l.o.g. we assume $\vr_i \not = \vec{0}$. Hence, from (\ref{eq2}) we get 
  $
  \vv\cdot (\pc- D \cdot \vr_i) > d
  $.
  By rearranging we get $\vv\cdot \pc-D\cdot \vv\cdot \vr_i>d$, and
  because $\vv\cdot \pc=d$, we get $D\cdot \vv\cdot \vr_i <0$ which
  contradicts that $\vv\cdot\vr_i \ge 0$. 
\end{proof}

Figures~\ref{geom3}.(b)-(c) show how to construct $\Hh$-closed polytope
from~(\ref{hpol}) for the \bmms{} in Figure~\ref{geom3}.(a), while
Figure~\ref{geom3}.(d) shows that for every corner of the constructed polytope
there is a mode that keeps the system inside the polytope.

The following lemma finally gives an algorithmically checkable
characterization of existence of $\Hh$-closed polytope.
\begin{lemma}\label{lem:extreme}
  Let $\Hh = (M, n, \Rr)$ be a \bmms{}.
  We have that for every $\vv$ there is a mode $m$ such that $\vv\cdot \vr \ge
  0$  for all $\vr \in m$ 
  if and only if $\Hh$ is safe. 
\end{lemma}
\begin{proof}
  In one direction, let us assume that $(M,n,R)\in\sem{\Ext(\Hh)}$
  is not safe, and let $Q=\{R(m)  \mid m\in M\}$.
  Then $\vec 0$ is not a convex combination of points in $Q$, and
  so by supporting hyperplane theorem applied to $\vec 0$ and the convex hull of 
  $Q$ there is $\vv$ and $d>0$ such that $\vv\cdot R(m) \ge d$ for all 
  $m\in M$. 
  Since $R(m)\in \Rr(m)$, this direction of the proof is finished.
  In the other direction,
  let $\vv$ be such that there is $\vr\in \Rr(m)$ for all $m\in M$ such that
  $\vv\cdot \vr<0$. 
  Then by convexity of $\Rr(m)$ there is $\vr_m\in \Vtx{\Rr(m)}$ with the same
  properties, and we can create  a \cmms{} $(M,n,R)\in\sem{\Ext(\Hh)}$ by
  putting $R(m)=\vr_m$. 
  This \cmms{} is not safe, because for any strategy, for a sufficiently large
  time bound a point $\px$ will be reached such that $(-\vv)\cdot \px$ is
  arbitrarily large, and hence any convex polytope will be left eventually. 
\end{proof}
\begin{algorithm}[t]
  \KwIn{\bmms{} $\Hh$, $\px\in \Real^n$ and $\gamma > 0$}
  \KwOut{$\Hh$-closed polytope $P$ contained in $\ball{\gamma}{\px}$ s.t. $\px\in P$, No if there is no $\Hh$-closed polytope.} 
  \ForEach{\cmms{} $H = (M, n, R)$ of $\sem{\Ext(\Hh)}$}{
    Check if there is a satisfying assignment for:
    \begin{eqnarray}
      \textstyle\sum\nolimits_{m \in M} R(m) \cdot t_m &=& \vec{0} \nonumber\\
      \textstyle\sum\nolimits_{m \in M} t_m &=& 1  \label{eqn1}\\
      t_m &\geq& 0 \text{ for all $m \in M$.}\nonumber
    \end{eqnarray}
    \lIf{no satisfying assignment exists}{\Return NO}
  }
  $R:=\{\vr_1,\vr_2,...,\vr_N\}$ be the set of rate vectors of $\sem{\Ext(\Hh)}$\;
  \Return the polytope given as convex hull of the points $\px + \frac{\gamma}{\sum_{i=1}^N \norm{\vr_i}} \cdot \sum_{i=1}^N \cdot p_i\vr_i$
  where $p_i \in \{0,1\}$\;   
  \caption{Schedulability Problem for Interior Starting States.}
  \label{SafePolytope}
\end{algorithm}
Combining Proposition~\ref{prop:dynsched} with Lemmas~\ref{lem:sched-pushing} and
\ref{lem:extreme} we get the following main result. 
\begin{theorem}
  \label{prop:helper-polytope}
  For every \bmms{} $\Hh$ and the starting state in the interior of a convex and
  bounded safety set we have that scheduler has a winning strategy if and only
  if $\Hh$ is safe. %
\end{theorem}
Theorem~\ref{prop:helper-polytope} allows us to devise
Algorithm~\ref{SafePolytope} and at the same time give its correctness. 
The reader may have noticed that Theorem~\ref{prop:helper-polytope} bears a striking
resemblance to Theorem~\ref{thmcms} for \cmms{}, since the former boils down to
checking safety of exponentially many \cmms{} instances. 
Note, however, that the proof here is much more delicate.
While in the case of \cmms{} satisfiability of~(\ref{eq:cmms}) gives immediately
a periodic winning strategy, for \bmms{} this is not the case:
even when every instance in $\sem{\Ext(\Hh)}$ is safe, we cannot immediately see
which modes should be used by the winning strategy; this requires the introduction of
$\Hh$-closed polytopes.

\subsection{General Case}
\label{sec:sched-general}
In this section we present Algorithm~\ref{SchedulabilityGamesGeneral} that
analyses schedulability of arbitrary starting states in $S$.
Notice that a starting state on the boundary of the safety polytope may lie on
various faces (planes, edges etc.) of different dimensions. 
The scheduler may have a winning strategy using modes that let the system stay
on some lower dimension face, or there may exists a winning strategy where scheduler
first reaches a face of higher dimension where it may have a winning
strategy.
Before we describe steps of our algorithm, we need to formalize a
notion of (open) faces of a convex polytope, a concept critical in
Algorithm~\ref{SchedulabilityGamesGeneral}.
\begin{algorithm}[t]
  \KwIn{\bmms{} $\Hh$, a safety set $S$ given by inequalities $A\vec{x}\le
    \vec{b}$, and a starting state $\px_0$.} 
  \KwOut{Yes, if the scheduler wins, No otherwise.}
  Compute the sequence  $\Ii = \seq{I_1, I_2, \ldots, I_\ell}$\;
  $\sched = \emptyset$,
  $\nosched = \emptyset$\;
  \ForEach{$I$ in $\Ii$}{ 
    \If{$I'\subseteq I$ and $I' \in \nosched$}{$\nosched := \nosched \cup \{I\}$\;}

    \If{$\exists m \in M$ with only internal rates}
    {$\sched := \sched \cup \{(I, \bot)\}$\;}%
    \Else{
      Construct $\Hh_{I}$\;
      \If{$\Hh_{I}$ is safe and $P_I$ is $\Hh_I$-closed polytope}{$\sched :=
        \sched \cup \{(I, P_I)\} $ ;}
      \lElse{$\nosched {:=} \nosched {\cup} \{I\}$\;}
    }
  }
  \lIf{$\exists I \in \sched$ and $\px_0 \models \facet{S}{I}$}
  {\Return Yes\;} 
  \lElse{\Return No\;}
  \caption{Schedulability Problem For Arbitrary Starting State}
  \label{SchedulabilityGamesGeneral}
\end{algorithm}

Let $Ax\le b$ be the linear constraints specifying a convex polytope $S$.
We specify a face of $S$ by a set $I\subseteq \{1,\ldots,\rows{A}\}$.
We write $\px\models \facet{S}{I}$, and we say that $\px$ satisfies
$\facet{S}{I}$, if and only if  $A_{1,j}x(1) + \cdots A_{n,j}x(n) = b_j$ for all
$j\in I$, and $A_{1,j}x(1) + \cdots A_{n,j}x(n) < b_j$ for all $j\not\in I$,
i.e. exactly the inequalities indexed by numbers from $I$ are satisfied tightly. 
Note that every point of $S$ satisfies $\facet{S}{I}$ for exactly one $I$.
Although there are potentially uncountably many states in every face of $S$
the following Lemma implies that it is sufficient to analyze only one state
in every face.
\begin{lemma}\label{lem:facets}
  For a \bmms{}, a convex polytope $S$, and for all faces $I$
  of $S$, either none or all states satisfying $\facet{S}{I}$ are schedulable.
  Moreover, if $I'{\subseteq} I$ and no point satisfying $\facet{S}{I'}$ is
  schedulable, then no point satisfying $\facet{S}{I}$ is schedulable. 
\end{lemma}
Let $\Ii= \seq{I_1, I_2, \ldots}$ be the sequence of all faces such that
$\facet{S}{I_i}$ is satisfied by some state, ordered such that if 
$I_i \subseteq I_j$, then $i\le j$. 
We call a mode $m$ {\em unusable for $I$} if there is $\px\models
\facet{S}{I}$ and  $\vr\in \Rr(m)$ such that $\px+\vr\cdot \delta \not\in S$ for
all $\delta > 0$. 
The rate $\vr$ satisfying this condition is called {\em external}. 
A rate $\vr$ is called {\em internal} if for any $\px$ such that $\px\models
\facet{S}{I}$ there is $\delta>0$ and $j$ such that $I_j\subseteq I$ and
$\px+\vr\cdot \varepsilon\models \facet{S}{I_j}$ for all $0<\varepsilon \le
\delta$. 
For a \bmms{} $\Hh$ and face $I$ we define a \bmms{} 
$\Hh_{I} = (M', n, \Rr')$ where $M'$ contains all modes of $M$ which are not
unusable for $I$,  and $\Rr'(m)$ is the set of non-internal rates of $\Rr(m)$. 

\begin{theorem}
  For every \bmms{} $\Hh$, a convex polytope safety set $S$, and a starting
  state $\px_0 \in S$, Algorithm~\ref{SchedulabilityGamesGeneral} decides
  schedulability problem for $\Hh$.
  Moreover, one can construct a dynamic winning strategy using the set $\sched$.
\end{theorem}
\begin{proof}{(Sketch.)}
Let $\seq{I_1, I_2, \ldots}$ be all sets such that $\facet{S}{I_i}$ is satisfied by
some state, ordered such that if $I_i \subseteq I_j$, then $i\le j$.
Algorithm~\ref{SchedulabilityGamesGeneral} analyzes the sets $I_i$, determining
whether the points satisfying $\facet{S}{I_i}$ are schedulable (in which case we
call $I_i$ schedulable), or not. 
Let us assume that $I$ is the first element of the sequence $\seq{I_1,I_2\ldots}$
which has not been analyzed yet. 
If there is $I'$ such that $I'\subseteq I$ and $I'$ is already marked as not
schedulable, then by Lemma~\ref{lem:facets} one can immediately mark $I$ as
non-schedulable. 
If all modes are unusable, then no point $\px$ such that $\facet{S}{I}$ is
schedulable. 
Notice that if there exists an internal rate to face $I_j$ then it must
necessarily be the case that $I_j$ is schedulable.
If there is a mode $m$ which only has internal rates, there
is a winning strategy $\sigma$ for the scheduler which starts by picking $m$
and a sufficiently small time interval $t$. 
This will make sure that after one step a point is reached which is already
known to be schedulable and scheduler has a winning strategy. 

If none of the previous cases match, the algorithm creates a \bmms{} $\Hh_{I}$
and applies Theorem~\ref{prop:helper-polytope} to the system $\Hh_{I}$.
If there is a $\Hh_I$-closed  polyhedron $P$, we know that $I$ is schedulable
and give a winning scheduler's strategy $\sigma_\px$ for any point $\px\models
\facet{S}{I}$ as follows.  
Let $d>0$ be a number such that for any $\py\models I_j$ where $j>i$ we have
$\norm{\px,\py} > d$, i.e. $d$ is chosen so that all points of $S$ contained in
$\ball{d}{\py}$ satisfy $\facet{S}{I'}$ for $I' \subseteq I$ (this follows from
the properties of the sequence $I_1,I_2,\dots$ and because $S$ is a convex
polytope). 
The strategy $\sigma_\px$ works as follows.
If all points in the history satisfy $\facet{S}{I}$, $\sigma_\px$ mimics
$\sigma_{\Hh_{I},\px,d}$.  
Otherwise, once a point $\py \not\models \facet{S}{I}$ is reached, the strategy
$\sigma_\px$ starts mimicking $\sigma_\py$. 
Note that the strategy $\sigma_\py$ is indeed defined by our choice of $d$ and
polytopes stored in $\sched$ set. %
Although the strategy we obtain in this way may potentially be non-positional,
it is a mere technicality to turn it into a positional one.

If $\Hh_{I}$ is not schedulable for any set and any point, then it is easy to
see that for no point satisfying $\facet{S}{I}$ there is a schedulable strategy. 
Indeed, for any strategy $\sigma$, as long as $\sigma$ picks the modes from
$M'$, the environment can play a counter-strategy showing that $\Hh_I$ is not
schedulable. 
When any mode from $m\in M\setminus M'$ is used by $\sigma$, we have that $m$ is
unusable and so the environment can pick a rate witnessing $m$'s unusability:
this will ensure reaching a point outside $S$.
Hence, we can mark $I$ as unschedulable. 
\end{proof}

\section{Complexity}
\label{sec:complexity}
In this section we analyze complexity of the schedulability problem for
\bmms{}. 
We begin by showing that in general it is co-NP-complete, however it can be
solved in polynomial time if the system has only two variables.

\subsection{General Case}
\begin{proposition}
  The schedulability problem for \bmms{} and convex polytope safety sets is in co-NP.
\end{proposition}
\begin{proof}[(Sketch)]
  We show that when the answer to the problem of schedulability of a point $\px$
  is No, there is a falsifier that consists of two components:
  \begin{itemize}
  \item 
    a set $I{\subseteq} \{1,\ldots, \rows{A}\}$ s.t. $\px \models
    \facet{S}{I'}$ for $I'\supseteq I$, and
  \item 
    a rate combination $(\vr_m)_{m\in M}$ such that there is a set of modes
    $\External\subseteq M$ where every $\vr_m$ for $m\in \External$ is external
    for $I$; and the rates $\vr_m$ for $m\not\in \External$ are neither external,
    nor internal, and there is a vector $\vv$ such that $\vv\cdot \vr_m > 0$
    for all $m\not\in \External$. 
  \end{itemize}
  
  Let us first show that the existence of this falsifier implies that the
  answer to the problem is No. 
  Indeed, as long as a strategy of a scheduler keeps using modes
  $m\not\in\External$, the environment can pick the rates $\vr_m$, and a point
  outside of $S$ will be reached under any non-Zeno strategy, because $S$ is
  bounded.  
  If the strategy of a scheduler picks any mode $m\in\External$, the environment
  can win immediately by picking the external rate $\vr_m$ and getting outside
  of $S$. 
  
  On the other hand, let us suppose that the answer to the problem is No, and
  let $I'$ be such that $\px\models\facet{S}{I'}$. 
  Then consider any \emph{minimal non-schedulable} $I\subseteq I'$.
  We put to $\External$ all modes which are unusable, and for every such mode,
  we pick a rate that witnesses it. 
  Further, there is not any mode with only internal modes and the \bmms{}
  $\Hh_I$ must be non-schedulable (otherwise $I$ would be schedulable, or would
  not be minimal non-schedulable). 
  By Proposition~\ref{prop:helper-polytope} there is an unsafe instance $H =
  (M', n, R) \in \sem{\Ext(\Hh_I)}$. 
  Since $M'$ contains all the modes whose indices are not in $\External$, we can
  pick the rate from this unsafe instance and we are finished. 
\end{proof}
\begin{figure}
  \centering
  {\small
  \begin{tikzpicture}[node distance=2cm,scale=0.45]
    \tikzstyle{lines}=[draw=black!30,rounded corners]
    \tikzstyle{vectors}=[-latex, rounded corners]
    \tikzstyle{cvectors}=[-latex, red, rounded corners]
    \tikzstyle{bvectors}=[-latex, blue, rounded corners]
    \tikzstyle{rvectors}=[-latex,very thick, rounded corners]

    \draw[lines,->] (-2.7,0)--(2.7,0);
    \draw[lines,->] (0, -2.7)--(0,2.7);
    \draw[lines,->] (-2.5, -2.5)--(2.5,2.5);
    
    \draw[vectors] (0, 0) --node[right]{$$} (1, 0) node[below, right]{$\vp_1$};
    \draw[vectors] (0, 0) --node[right]{$$} (0, 1)node[above]{$\vp_2$};
    \draw[vectors] (0, 0) --node[right]{$$} (-0.7, -0.7)node[left]{$\vp_3$};
    \draw[cvectors] (0, 0) --node[right]{$$} (-1, 0)node[left,above]{$\vn_1$};
    \draw[cvectors] (0, 0) --node[right]{$$} (0 ,-1)node[below]{$\vn_2$};
    \draw[cvectors] (0, 0) --node[right]{$$} (0.7, 0.7)node[ right]{$\vn_3$};

    \draw[lines,->] (7.3,0)--(12.7,0);
    \draw[lines,->] (10, -2.7)--(10,2.7);
    \draw[lines,->] (7.5, -2.5)--(12.5,2.5);
    \draw[vectors] (10, 0) --node[right]{$$} (11, 0) node[below]{$\vp_1$};
    \draw[bvectors] (10, 0) --node[right]{$$} (11, 0) node[above]{$\vp_1$};
    \draw[bvectors] (10, 0) --node[right]{$$} (10, 1)node[above]{$\vp_2$};
    \draw[bvectors] (10, 0) --node[right]{$$} (9.3, -0.7)node[left]{$\vp_3$};
    \draw[cvectors] (10, 0) --node[right]{$$} (9, 0)node[left,above]{$\vn_1$};

     \end{tikzpicture}
}
  \caption{An example from proof of Proposition~\ref{prop:mmms-hard}}
  \label{reduction}
\end{figure}
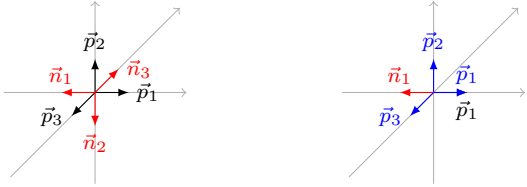
\vspace*{-1em}
\begin{proposition}[co-NP hardness]\label{prop:mmms-hard}
  The schedulability problem for \mmms{} is co-NP hard. 
\end{proposition}
\begin{proof}[(Sketch)]
  The proof for co-NP hardness uses a reduction from the classical NP-complete
  problem $3$-SAT.
  For a SAT instance $\phi$ we construct a \mmms{} $\Hh_\phi$ such that $\phi$
  is satisfiable if and only if $\Hh_\phi$ is not schedulable for any starting
  state and bounded convex safety set. 
  We only sketch the construction of $\Hh_\phi$ here and formally prove the
  correctness of the construction in
\ifthenelse{\isundefined{\techreport}}{%
\cite{techreport}}{%
 Appendix~\ref{app:mmms-hard}}.
  Consider a SAT instance $\phi$ with $k$ clauses and $n$ variables denoted as 
  $x_1, \ldots, x_n$.  
  The corresponding \mmms{} $\Hh_\phi = (M, n, \Rr)$ is such that its set of
  modes $M=\{m_1,\ldots,m_k\}$ corresponds to the clauses in $\phi$,
  and variables are such that variable $i$ corresponds to variable $x_i$
  of $\phi$. 
  For each variable $x_i$ we define vectors $\vp_i$ and $\vn_i$
  such that $\vp_i(i) = 1$, $\vn_i(i) = -1$, and $\vp_i(j) = \vn_i(j) = 0$ if
  $i\not = j$. 
  The rate-vector function $\Rr$ is defined such that for each mode $m_j$ and
  for each SAT variable $x_i$ we have that $\vp_i \in \Rr(m_j)$ if $x_i$ occurs
  positively in clause $j$, and $\vn_i \in \Rr(m_j)$ if the variable $x_i$ occurs
  negatively in clause $j$.  
  The crucial property here is that there is no vector that can have a
  positive dot product with both $\vp_i$ and $\vn_i$, which allows us to
  map unsafe rate combinations to satisfying valuations and vice versa. 
    Figure~\ref{reduction} shows an example of the reduction for two different
  formulas. 
  On the left, we have a satisfiable formula
  $(x_1 \vee x_2 \vee x_3) \wedge (\neg x_1 \vee \neg{x}_2 \vee
  \neg{x}_3)$ which gives rise to a \mmms{} with two modes: 
  $\set{\vp_1,\vp_2,\vp_3 } \in m_1$ and 
  $\set{\vn_1, \vn_2, \vn_3} \in m_2$. %
  The system has unsafe combination $\vp_1, \vn_2$. %
  In Figure~\ref{reduction}~(right) an unsatisfiable formula
  $(x_1 \vee x_1 \vee x_1) \wedge (\neg{x}_1 \vee \neg{x}_1 \vee \neg{x}_1) \wedge (x_1 \vee x_2 \vee x_3)$ 
  is reduced to a \mmms{} with three modes:
  $\set{\vp_1} \in m_1$, $\set{\vn_1} \in m_2$, and $\set{\vp_1, \vp_2, \vp_3} \in m_3$. 
  All combinations are safe.
\end{proof}
The proof of the following easy corollary is postponed to
\ifthenelse{\isundefined{\techreport}}{%
\cite{techreport}}{%
Appendix~\ref{app:bmms-hard}}. 
\begin{corollary}[co-NP hardness result for \bmms{}]\label{cor:bmms-hard}
  The schedulability problem for \bmms{} is co-NP hard.
\end{corollary}

\subsection{\bmms{}  with two variables}
\begin{algorithm}[t]
  \KwIn{\bmms{} $\Hh$ with two variables.}
  \KwOut{Return Yes, if $\Hh$ is safe and No otherwise.}
  Set $R$ to the set of extreme rate vectors of $\Hh$\; 
  \ForEach{$\vr_\bot \in R$ }{
	Set $\vec u$ to be a perpendicular vectors to $\vr_\bot$\;
	\ForEach{$\vv \in \{\vec u, -\vec u\}$}{
    		\lIf{for each $m\in M$ there is $\vr \in m$ s.t. $\vv \cdot \vr >0$ or there is $p>0$ s.t. $\vr = p \vr_\bot$}{
			\Return No;
		}
	}
    }
    \Return Yes
  \caption{Decide if a two dimension \bmms{} is safe.}
  \label{SafePolytopeTwo}
\end{algorithm}
For a special case of \bmms{} which only have two variables, we show the
following result. 
\begin{theorem}
  Schedulability problems for \bmms{} with convex polytope safety sets
  are in P for systems with $2$ variables.   
\end{theorem}
The rest of the section is devoted to the proof of this theorem.
The following lemma shows that to check whether a set of rate vectors
$R=\set{\vec r_1,...,\vec r_k}$ is unsafe it is sufficient to check properties 
of vectors $\vu$ perpendicular to some vector of $R$. 
This observation yields a polynomial time algorithm.
\begin{lemma}
 \label{unsafeCombo}
 Let $R$ be a set of vectors. 
 There is $\vv$ such that $\vv\cdot \vr>0$ for all $\vr\in R$ if and only if 
 there are $\vec u$ and $\vr_\bot\in R$ satisfying $\vec u\cdot \vr_\bot = 0$
 and for all $\vr\in R$ either $\vec u\cdot \vr > 0$ or $\vr=p\cdot\vr_\bot$ for
 some $p>0$. 
\end{lemma}
\begin{proof}[(Sketch)]
 To obtain $\vv$ we keep changing $\vv$ until it becomes perpendicular to some 
 vector in $R$. On the other hand, $\vv$ is obtained from $\vu$ by making a
 sufficiently small change to $\vu$. 
 A formal proof is given in
\ifthenelse{\isundefined{\techreport}}{%
\cite{techreport}}{%
Appendix~\ref{app:unsafeCombo}}.
\end{proof}

\begin{figure}[b]
  \centering
  {\small
  \begin{tikzpicture}[node distance=2cm,scale=0.45]
    \tikzstyle{lines}=[draw=black!30,rounded corners]
    \tikzstyle{vectors}=[-latex, rounded corners]
    \tikzstyle{rvectors}=[-latex,very thick, rounded corners]

    \draw[lines,->] (-3,0)--(3,0);
    \draw[lines,->] (0, -3)--(0,3);
    \draw[vectors] (0, 0) --node[right]{$$} (1, -1) node[below,left]{$\vr_1$};
    \draw[vectors] (0, 0) --node[right]{$$} (2, -2)node[right]{$\vr_2$};
    \draw[vectors,dashed] (0, 0) --node[right]{$$} (1, 0)node[right]{$\vv$};
    \draw[vectors,dashed] (0, 0) --node[right]{$$} (2, 2)node[above]{$\vu$};
    \draw[vectors] (0, 0) --node[right]{$$} (2, 1)node[right]{$\vr_3$};
    \draw[vectors] (0, 0) --node[right]{$$} (0.5, 2)node[above]{$\vr_4$};

      \draw[lines,dashed] (-3,3)--(3,-3);

    \draw[lines,->] (7,0)--(13,0);
    \draw[lines,->] (10, -3)--(10,3);
    \draw[vectors] (10, 0) --node[right]{$$} (11, 1) node[below,left]{$\vr_1$};
    \draw[vectors] (10, 0) --node[right]{$$} (9, 1)node[right]{$\vr_2$};
    \draw[vectors] (10, 0) --node[right]{$$} (10, -1)node[right]{$\vr_3$};
      \draw[lines,dashed] (7,-3)--(13,3);
      \draw[lines,dashed] (7,3)--(13,-3);

     \end{tikzpicture}
}
  \caption{Examples for Lemma \ref{unsafeCombo}} %
  \label{lemmaExample}
\end{figure}
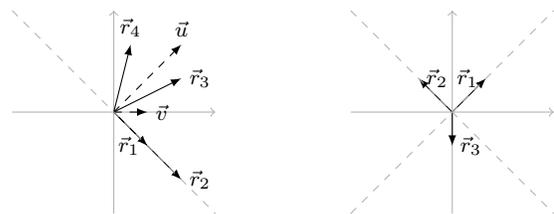

Consider an unsafe set of rate vectors $R = \set{\vr_1,\vr_2,\vr_3,\vr_4}$
shown in Figure \ref{lemmaExample}~(left).
All the rate vectors are on the right side of line $y=0$ and vector $\vv$  
has strictly positive dot product with all of them. 
As it can be seen in the figure, all the rate vectors are on right-hand side of
the line passing through $\vr_1$ and there exists $\vu$ perpendicular to $\vr_1$
such that $\vv' \cdot \vr_i \geq 0$ for all $\vr_i \in R$.  
Observe that adding a rate vector $\vr_5=-\vr_1$ to $R$ makes this set of rate
vectors safe, and none of rate vectors would satisfy the conditions of Lemma
\ref{unsafeCombo}. 
Figure \ref{lemmaExample} (right) shows a safe set of rate vectors. 
As one can see none of rate vectors has the others on one side of itself.
The following corollary implies that we can use Lemma \ref{unsafeCombo} to check
the safety of a \bmms{}. 
\begin{corollary}
  A \bmms{} $\Hh$ with two variables is not safe if and only if there
  exists a rate vector $\vr_\bot$ in one of the modes of system and vector $\vv$
  perpendicular to it, such that for all modes $m \in \Hh$: (i) there exists
  $\vr \in m$ such that $\vv \cdot \vr > 0$; or (ii) $\vv \cdot \vr = 0$ and
  $\vr=p \cdot \vr_\bot$ for some $p>0$.  
\label{combos}
\end{corollary}

Algorithm \ref{SafePolytopeTwo} checks whether all the combinations are safe in
polynomial time; it chooses a rate vector $\vr_\bot$ at each step and tries to
find an unsafe combination using the result of Corollary \ref{combos}. Note that
for any non-zero vector $\vr_\bot$ in two dimensions there are only two vectors
which we need to check.  
Although there are infinitely many vectors $\vv$ which might satisfy conditions
of Corollary~\ref{combos}, the conditions we are checking are preserved if we
multiply $\vv$ by a positive scalar.

\section{Discrete Schedulability}
\label{sec:discrete}
In this section we discuss the {\em discrete schedulability problem}, in which
a scheduler can only make decisions at integer multiplies of a specified \emph{clock
period} $\Delta$ and the environment has finitely many choices of rates.
Formally, given a \mmms{} $\Hh$, a closed convex polytope $S$ as safety set, an
initial state $x_0 \in S$, the discrete schedulability problem is to decide if
there exists a winning strategy of the scheduler where the time delays are
multiples of $\Delta$.   

\begin{theorem}\label{thm:discr}
  Discrete schedulability problem is EXPTIME-complete.
\end{theorem}
\begin{proof}
  EXPTIME-membership of the problems is shown via discretization of the state
  space of $\Hh$.
  Since the set $S$ is given as a bounded polytope, the size of the
  discretization can be shown to be at most  exponential in the size of $\Hh$
  and $\Delta$, and since the safety games on a finite graph can be solved in P,
  EXPTIME membership follows.
  The hardness can be proved by a reduction from the countdown
  games~\cite{JLS08}. For space constraints we give the proof in
\ifthenelse{\isundefined{\techreport}}{%
\cite{techreport}}{%
  Appendix~\ref{app:delta-hard}}.
\end{proof}

We turn the discrete schedulability problem to an optimization problem, by asking
to find supremum of all $\Delta$ for which the answer to the discrete
schedulability problem is yes. 
We prove the following, which also improves a result of~\cite{ATW12} where only an
approximation algorithm was given.

\begin{theorem}
  \label{thm:max-delta}%
  Given a \mmms{} $\Hh$, a closed convex polytope $S$ and an initial state $\px_0$,
  there is  an exponential time algorithm which outputs the maximal $\Delta$ for
  which the answer to the discrete schedulability problem is Yes. 
  For a \cmms{} the algorithm can be made to run in polynomial space.
\end{theorem}
\begin{proof}[(Sketch)]
  We exploit the fact that as the clock period $\Delta$ increases, all the
  points of the discretization move continuously towards infinity, except for
  the initial point. 
  This further implies that for $\Delta$ to be maximal, there must be a point of
  the discretization which lies on the boundary of $S$, since otherwise we could
  increase $\Delta$ by some small number, while preserving the existence of a
  safe scheduler. 
  By using a lower bound on $\Delta$ from Section~\ref{sec:schedulability}
  (obtained as a by-product of the construction of a dynamic strategy), there
  are only exponentially many candidates for such points, which gives us
  exponentially many candidates for maximal $\Delta$ to consider, and we can
  check each one by Theorem~\ref{thm:discr}.
  For the PSPACE bound we don't enumerate the points, but guess them
  nondeterministically in polynomial space, and utilize~\cite[Theorem 10]{ATW12}
  instead of Theorem~\ref{thm:discr}. 
  Full details of the proof are given in
\ifthenelse{\isundefined{\techreport}}{%
\cite{techreport}}{%
Appendix~\ref{app:delta-max}}.
\end{proof}

\section{Conclusion}
We investigated systems that comprise finitely many real-valued
variables whose values evolve linearly based on a rate vector determined by
strategies of the scheduler and the environment.
We studied an important schedulability problem for these systems, with
application to energy scheduling, that asks whether scheduler can make sure that
the values of the variables never leave a given safety set.
We showed that when the safety set is a closed convex polytope, existence of
non-Zeno winning strategy for scheduler is decidable for any arbitrary starting
state.  
We also showed how to construct such a winning strategy. 
On complexity side, we showed that the schedulability problem is co-NP complete in
general, but for the special case where the system has only two variables, the
problem can be decided in polynomial time. 
Directions for future research include investigation of
schedulability problem with respect to more expressive higher-level control
objectives including temporal-logic based specification and bounded-rate
multi-mode systems with reward functions.

\bibliographystyle{plain}
\bibliography{papers}

\ifthenelse{\isundefined{\techreport}}{}{
\newpage

\appendix
\section{Absence of Static Strategies}\label{app:static}
\begin{proposition}
  \label{Prop:subsystem}
  For a given starting state in the interior of the safety set $S$, the
  scheduler has a static winning strategy in a \bmms{} $\Hh = (M, n, \Rr)$ iff there is
  $M'\subseteq M$ such that $|\Rr(m)|=1$ for all $m\in M'$, and
  the \cmms{} $H=(M',n,R)$ is safe, where $R(m)$ is the unique rate of $\Rr(m)$.
\end{proposition}
\begin{proof}
   The ``if'' direction is trivial. 
  To show the ``only if'' direction we show that if there is no \cmms{} subsystem of $\Hh$ for which  there is a safe and non-Zeno schedule, then
  there is no static winning schedule for schedulability  objective.
  Let $\sigma = (m_1, t_1), (m_2, t_2), \ldots$ be a static scheduler.
  
  Assume there is
  $m\in M$ with two different rates $\vr_a$ and $\vr_b$ such that $\sum_{i: m_i=m} t_i = \infty$.
  We then define two strategies for the environment, $\pi_a$ and $\pi_b$ which for
  a mode $m$ always pick a rate $\vr_a$ and $\vr_b$, respectively. After the first $k$
  steps, the point reached under $\sigma_b$
  is equal to
  \[
   \px_b = \px_a + (\vr_b - \vr_a) \cdot \sum_{i\le k : m_i = m} t_i
  \]
  Hence, the points $\px_a$ and $\px_b$ will be arbitrarily far apart for large enough $k$,
  since the safety set is bounded, one of the strategies  $\pi_a$ and $\pi_b$ must ensure
  that a point outside is left eventually.

  On the other hand, assume all modes $m$ which have two different rates satisfy that
  $\sum_{i: m_i=m} t_i$ is finite.
  Let $M'$ be all such modes, and let
  $d_1 := \norm{\vr}\cdot \sum_{i: m_i\in M'} t_i \le \infty$ where $\vr$ is the rate with
  the maximal Euclidean norm which occurs in $\Hh$. Intuitively, $d$ is the upper bound on the change of the values
  of variables caused by using the modes of $M'$.
  Let $d_2$ be the diametre of $S$, and let $p$ be the Euclidean distance of 
  the initial point $\px_0$ from the boundary of $S$. We define a strategy
  \[
   \sigma' = (m'_1, t'_1\cdot \frac{p}{d_1+d_2}), (m'_2, t'_2\frac{p}{d_1+d_2}), \ldots
  \]
  where $(m'_1, t'_1), (m'_2, t_2), \ldots$ is the sequence $(m_1, t_1), (m_2, t_2)\ldots$
  from which we omit all the tuples which have a mode from $M'$ in the first component.
  The strategy $\sigma'$ is safe for $S$ and further shows that there is a safe
  \cmms{} subsystem, which is a contradiction. 
\end{proof}

%

%
%
\section{Omitted Proofs}

\subsection{Proof of Lemma~\ref{lem:facets}}
  Let $\px$ and $\py$ be points satisfying $\facet{S}{I}$.
  Assume $\px$ is safe with a strategy $\sigma$, and let $\sigma_d$ be a strategy for a controller defined as follows:
  Let $\varrho=\seq{\py_0, (m_1, t_1), \vr_1, \py_1, \ldots \py_k}$
where $\py_i = \py + \py'_i$ for some $\py'_i$,
be a history, and let $(\vr,t)$ be a decision of
$\sigma$ on $(\px_0, m_1, t'_1, \vr_1,\px_1, \ldots \px_k)$,
where $\px_i = \px + \py'_i\cdot d$ and $t'_i = t_i\cdot d$. The
strategy $\sigma_d$ chooses $(\vr,t/d)$ in $\pi$. Intuitively, $\sigma_d$ mimics
the decision of $\sigma$, but it assumes the starting point is $\py$
rather than $\px$, and it scales the time intervals down by $d$,
hence making sure that only points closer to $\py$ can be reached. For this reason
it suffices to take large enough $d$ to make sure that $\sigma_d$ is safe. For example,
we can put $d=(\sup_{\px'\models \facet{S}{I}} \norm{\px,\px'})/(\inf_{\py'\models \facet{S}{I'}, I\subseteq I'} \norm{\py,\py'})$.

Similar arguments can be made for the second part of the lemma, i.e. any
strategy safe for a point satisfying $\facet{S}{I'}$ can be scaled to a strategy
safe for a point satisfying $\facet{S}{I'}$.

\subsection{Proof of Proposition~\ref{prop:mmms-hard} (correctness of construction)}\label{app:mmms-hard}
We show that the construction proposed in the proof of Proposition~\ref{prop:mmms-hard}
in the main body is correct.
  We show that there is a satisfying assignment
  for $\varphi$ iff there exists an unsafe instance of $\Hh_\phi$. 
  \begin{itemize}
  \item
    Now let us suppose that there is an unsafe combination $\{\vec r_i | \vec r_i
  \in m_i, \: 1 \le i  \le |M|\}$. 
  Then for every rate $\vr_i$ which contains $1$ at $i$-th position assign true to
  the variable $x_i$, and for every rate $\vr_i$ which contains $-1$ at $i$-th
  position assign false to the variable $x_i$. 
  Note that no variable would be assigned both true and false since if two
  vectors $\vr'$ and $\vr''$ are chosen which go to the opposite direction, then
  every $\vv$ which satisfies $\vv\cdot r' >0$ also satisfies $\vv\cdot r'' < 0$, and
  vice versa, which means that the combination is not unsafe.
 Further, observe that the assignment is satisfying, because for every clause
  $c_j$ we have that if $\vr_j$ contains $1$ at $i$-th position, then $c_j$
  contains the literal $x_i$ which is satisfied, and if $\vr_j$ contains $-1$ at
  $i$-th position, then $c_j$ contains the literal $\neg x_i$ which is
  satisfied. 
  Hence there is at least one true literal in each clause and thus the formula
  $\phi$ is satisfiable.  

  \item
    To prove the other direction, assume that there is a satisfying assignment to
  $\phi$, then choose one true literal from each clause and consider the
  corresponding rate vector for each mode. 
  Note that there would be no two vectors along one axis with different
  directions since $\neg x_i$ and $x_i$ can not be true at the same time.
  Therefore we have $k$ vectors along $1 \le d \le n$ axises where each two
  vectors are either same or perpendicular.
  This set of rate vectors will be unsafe since there exists a $\vec v$ with
  strictly positive dot product with all of them:
  We build vector $\vec v$ such that each $i$-th entry of vector $\vec v$ is $1$
  (resp. $-1$), if there are some vectors whose $i$-th entry is $1$
  (resp. $-1$),  and zero otherwise.
  The product of $v$ with any vector from the combination is equal to $1$, and
  hence greater than zero. \qed
\end{itemize}

\subsection{Proof of Corollary~\ref{cor:bmms-hard}}\label{app:bmms-hard}
To prove this corollary we show that if there is an unsafe instance of \bmms{}
$\Hh$ then there is an unsafe instance of corresponding extreme-rate \mmms{}
$\Ext(\Hh)$. 
With this observation, the corollary then follows from Proposition~\ref{prop:mmms-hard}.
Assume $m$ is a mode in the bounded-rate multi-mode system $\Hh$ with extreme
rate vectors $\{r^*_1,...,r^*_k\}$.  
First we show that if there is a rate vector $r \in m$ and a rate vector $v$
such that their dot product is positive, i.e. $v.r>0$, then there exists at
least one extreme rate vector $r^*_i$ which makes angle less than 90 with $v$,
i.e. $v.r^*_i>0$.  
We can write $r=\sum \lambda_i r^*_i$ where $\sum \lambda_i =1$. Assume vector
$v$ has positive dot product with $r$, $v.r>0$. Assume for the purpose of
contradiction that $\forall i ~ v.r^*_i \leq 0$, which is a contradiction
because then we have $v.\lambda_ir^*_i \leq 0 \rightarrow \sum
v.\lambda_ir^*_i \leq 0 \rightarrow v.\sum \lambda_i r^*_i = v.r \leq 0$.  
Thus if there is an unsafe instance of \bmms{}, for each mode we can choose a
extreme rate such that the corresponding extreme-rate instance is unsafe.  

\subsection{Proof of Lemma~\ref{unsafeCombo}}
\label{app:unsafeCombo} 
If $|R|\le 1$, then the claim is immediate. Assume $R$ contains at least two
rates. 

Let us start with $\Rightarrow$. 
Intuitively, we keep changing $\vv$ until it becomes perpendicular to some
vector in $R$, and then we show that the vector obtained in this way satisfies
the desired properties. 
Formally, pick a vector $\vw$ such that $\vw\cdot \vr_0=0$.
Find a maximal $\alpha\in [0,1)$ such that for the vector $\vv_\alpha :=
\alpha\cdot\vv + (1-\alpha)\cdot \vw$ there is a vector in $R$ perpendicular to
$\vv_\alpha$. Such $\alpha$ must exists, since at least for $\vv_0 = \vw$ we
have $\vr_0$ perpendicular.
We claim $\vv_\alpha$ is our vector $\vu$, and we put $\vr_\bot$ any vector of
$R$ perpendicular to it. 
First, observe that there is no $\vr\in R$ such that $\vu\cdot \vr <0$. If this 
was the case, then $\alpha\cdot\vv\cdot \vr + (1-\alpha)\cdot \vw\cdot \vr <0$
and since $\vv\cdot \vr$ is positive, we could have picked $\alpha' > \alpha$
for which $\vv_{\alpha'} \cdot \vr = \alpha'\cdot\vv\cdot \vr + (1-\alpha')\cdot
\vw\cdot \vr = 0$ (for the same $\vr$ as before), contradicting the maximality
of $\alpha$. 
Now for any vector $\vr\in R$ such that $\vu\cdot \vr = 0$, if $\vr\neq
p\vr_\bot$ for any $p>0$, then $\vr= p\vr_\bot$ for some $p<0$. 
But since $\vr_\bot \cdot \vv >0$, we get $\vr\cdot \vv = p\cdot (\vr_\bot \cdot
\vv) < 0$, which is a contradiction with properties of $\vv$.
 
In the other direction, if there are no $\vr\in R$ such that $\vec u\cdot \vr >
0$, we can just put $\vv$ to be an arbitrary element of $R$.
Otherwise, we show that we can obtain $\vv$ if we make a small enough change to
$\vu$. 
Fix some $\vr_\bot$ where $\vu\cdot \vr_\bot=0$. 
Let $\tau := \min_{\vr\in R: \vu\cdot \vr >0} \vu\cdot\vr$ be the minimal
\emph{positive} dot product of $\vu$ with vectors of $R$, and let $\kappa :=
\min_{\vr\in R} \vr_\bot \cdot \vr$ be the minimal (possibly negative) dot
product of $\vr_\bot$ with vectors of $R$. 
Set $\vv = \vu + \frac{\tau}{2\cdot (|\kappa|+1)} \vr_\bot$. 
For every $\vr\in R$, we have 
$\vv\cdot \vr = \vu\cdot \vr + \frac{\tau}{2\cdot (|\kappa|+1)} \vr_\bot \cdot
\vr$ which is positive, because: 
(i) if the left summand is $0$, then the right summand is positive because
$\vr_\bot \cdot \vr>0$, and 
(ii) if the left summand is positive, then it is at least $\tau$ and
the right summand is at least $\frac{\tau}{2\cdot(|\kappa|+1)} \kappa\ge
-\frac{\tau}{2}$, and so the sum is positive. 

\subsection{Proof of Theorem~\ref{thm:discr} (the hardness part)}\label{app:delta-hard}
A countdown game is a tuple $\mathcal{G} = (N, T, n_1, B_1)$ where 
\begin{itemize}
\item
  $N = \set{n_1, n_2, \ldots, n_d}$ is a finite set of nodes;
  \item
    $T \subseteq N \times \Nat_{>0} \times N$ is a set of transition; and
  \item
    $(n_1, B_1) \in N {\times} \Nat_{>0}$ is the initial configuration.
  \end{itemize}
  From any configuration $(n, B) \in N \times \Nat_{>0}$, 
  first player~1 chooses a number  $k \in \Nat_{>0}$, such that $k
  {\leq} B$ and there exists some $(n, k, n') \in T$, and then player~2
  chooses a transition $(n, k, n'') \in T$ labeled with that number. 
  Note that there can be more than one such transition.
  The new configuration then transitions to $(n'', B {-} k)$. 
  Player~1 wins a play of the game when a configuration $(n, 0)$ is
  reached, and loses (i.e., player~2 wins) when a configuration
  $(n, B)$ is reached in which player~1 is stuck, i.e., for all transitions
  $(n, k, n') \in T$, we have $k > B$.
  
  For a countdown game $(N, T, n_1, B_1)$ we define a \bmms{} $\Hh$, a safety
  set $S$ and an initial state $\px$ such that there is a safe scheduler in
  $\Sigma_\Delta$ for $\Delta=1$ iff player 1 has a winning strategy in the
  countdown game.   
  W.l.o.g we assume that when $(n, k, n')\in T$, then $n \neq n'$, and also we
  assume that the initial state is $(n_1, B_1)$ and there is no node $n$ and $k$
  such that $(n, k, n_1)\in T$. 

  The \bmms{} $\Hh$ has $d+1$ variables. 
  The intuition is that the value of the first variable corresponds to the value
  of the counter, while $(i+1)$th variable is equal to $1$ if the game is in
  node $n_i$, and $0$ otherwise. 
  
  For all $n, k \in N \times \Nat_{>0}$ such that there is $(n, k, n')\in T$ for
  some $n'$, we add a  mode $(n, k)$ to $\Hh$. 
  For all $(n_i, k, n_j)\in T$, we add the rate $r$ to the mode $(n_i, k)$
  such that the first component of $r$ is $-k$, the $(i+1)$th component is
  $-1$ and $(j+1)$th component is $1$. 
  All other components of $r$ are zero. 
  We further add modes $m_i$ for $3 \le i \le d+1$ which contain the unique rate
  with $B_1$ in the first component, $1$ in the second component, and $-1$ in
  $i$-th component. 
  All other components of this rate are zero. 

  The safety set $S$ is defined so that the only points with integer values are
  exactly $(i_1,\ldots, i_{d+1})$, where $0\le i_1\le B_1$, and exactly one of
  $i_2,\ldots, i_{d+1}$ is $1$, while the others are $0$. Such safety set can be
  defined using equations 
  \[
  \begin{array}{rl}
    x_1&\le B_1\\
    \sum_{i=2}^{d+1} x_i &\le 1\\
    x_i &\ge 0\qquad\text{for $1\le i\le d+1$}
  \end{array}
  \]
  Now we claim that the system is schedulable from the point
  $(B_1,1,0,0,\ldots,0)$ iff player 1 has a winning strategy in the countdown
  game. 
  The intuition is that the winning strategy for player 1 in the countdown game
  directly gives a strategy for the scheduler in $\Hh$ such that a point is
  reached which has zero in the first component, and zeros everywhere else
  except for some $i$-th component. 
  Then the scheduler uses the mode $m_i$, which leads back to the initial state
  and then he can repeat the same strategy. 
  On the other hand, if player 2 has a winning strategy in the countdown game,
  this strategy can be used to get to a state from which the scheduler has no
  chance but to leave the safety set (which corresponds to not having any
  choices in the countdown game).  

\subsection{Proof of Theorem~\ref{thm:max-delta}}
\label{app:delta-max}
In this section we show how to solve the following problem: given a \mmms{}
$\Hh = (M, n, \Rr)$, a convex polytope $S$ and an initial state $\px\in S$, find the maximal number $\Delta_{max}$
such that there is a winning strategy for the scheduler which only takes decisions at times $i\cdot \Delta_{max}$ where $i\in\Nat$.
Formally, let $\Sigma_\Delta$ denote the set of strategies for the scheduler which schedule in multiples of $\Delta$. Then
we wish to find a supremum, over all $\Delta$, such that there is a safe scheduler in $\Sigma_\Delta$.

Let $R$ be the set of all possible rate vectors of $\Hh$. Note that since $\Hh$ is a \mmms{}, the set $R$ is finite.

Let $\discr{\Delta}$ be the points reachable from $\px$ when using a scheduler from $\Sigma_\Delta$.
All such points are equal to $\px + \sum_{\vr\in R} i_{\vr} \cdot \Delta \cdot \vr$ for some $i_{\vr}\in \Nat$.
This implies that the set $\discr{\Delta}\cap S$ is finite.

The intuition of our algorithm is the following. Every strategy from $\Sigma_\Delta$ can be seen as a function which rather
than observing and choosing time delays observes and chooses the number of time periods (multiples of $\Delta$) elapsed. Using this
abstracted view of strategies, every strategy in $\Sigma_\Delta$ corresponds to a strategy in $\Sigma_\Delta'$ which differs only
in the length of the time period. It can be shown that there is a correspondence of points reachable under these two strategies. Seeing the
points of $\discr{\Delta}$ as a ``grid'', the points of $\discr{\Delta'}$ are obtained by stretching (if $\Delta'>\Delta$) or squeezing (if $\Delta'<\Delta$) this grid.
It follows that for a $\Delta$ to be maximal, there must be a point in $\discr{\Delta}$ which lies on the boundary of $S$, since otherwise
the grid $\discr{\Delta}$ can be stretched to some $\discr{\Delta'}$ where $\Delta'>\Delta$, preserving the existence of a safe scheduler.
Exploiting this property together with the fact that we already know a lower bound on $\Delta_{max}$,
we get only finitely many candidates for maximal $\Delta$, and we can check in each of them whether 
a safe scheduler exists using Theorem~\ref{thm:discr}.
Our algorithm is presented as Algorithm~\ref{algmax}.
\begin{algorithm}
\KwIn{schedulable $\Hh$, safety set $S$ given as $Ax\le b$, point $\px\in S$}
\KwOut{$\Delta_{max}$}
Let $\Gamma$ be the lower bound on $\Delta_{max}$\;
Compute $\discr{\Gamma}\cap S$\;
$\Delta_{max} := \Gamma$\;
\ForEach {$\py = \px + \sum_{\vr\in R} i_{\vr} \Delta \cdot \vr \in \discr{\Gamma}\cap S$}{
\nllabel{line:lp}maximise $\Delta$ subject to
$A \cdot (\px + \sum_{\vr\in R} i_{\vr} \Delta \cdot \vr) \le b$
\If{$\Sigma_{\Delta}$ contains a safe scheduler and $\Delta > \Delta_{max}$}{
$\Delta_{max} := \Delta$
}
\Return $\Delta_{max}$
}
\caption{algorithm computing $\Delta_{max}$}
\label{algmax}
\end{algorithm}

Let us now prove the correctness of the algorithm.
Clearly the algorithm terminates in exponential time since the ``foreach'' loop
is executed only exponentially many times at most, and each of the respective lines can be executed in exponential time.
Hence, we only need to show that the result returned by the algorithm is correct. 

We first introduce some technical notation to capture the intuition of correspondence between points of different discretisations.
Define a bijection $g_{\Delta, \Delta'}$ between $\discr{\Delta}$ and $\discr{\Delta'}$
that to a point $\px + \Delta\cdot \sum_{\vr\in R} i_{\vr} \cdot \vr$ where $i_{\vr}\in \Nat$
assigns the point $\px + \Delta' \cdot \sum_{r\in R} i_{\vr} \cdot \vr$. Intuitively, this function
pairs the corresponding points on the ``grids'' given by $\discr{\Delta}$ and $\discr{\Delta'}$.
Note that $g_{\Delta, \Delta'}$ is well defined and does not depend on the choice of $i_{\vr}\in \Nat$ which represent the point and can be non-unique.
Indeed, if
\[
 \px + \Delta\cdot \sum_{\vr\in R} i_{\vr} \cdot \vr = \px + \Delta\cdot \sum_{\vr\in R} i'_{\vr} \cdot \vr
\]
for some $i_{\vr},i'_{\vr}\in \Nat$, then $\sum_{\vr\in R} i_{\vr} \cdot \vr = \sum_{\vr\in R} i'_{\vr} \cdot \vr$
and hence also $\px + \Delta'\cdot \sum_{\vr\in R} i_{\vr} \cdot \vr = \px + \Delta'\cdot \sum_{\vr\in R} i'_{\vr} \cdot \vr$.

The following lemma essentially says that when we enlarge the length of a time period, the set of points on the corresponding grid that
are within $S$ can only get smaller.
\begin{lemma}\label{lemma:points-monotone}
Let $\Delta \ge \Delta'$. Then $g_{\Delta',\Delta}(\discr{\Delta'}\cap S) \supseteq \discr{\Delta}\cap S$.
\end{lemma}
\begin{proof}
Follows because $S$ is closed, convex and contains $\px$.
\end{proof}

The following lemma intuitively says when we can increase the time period while preserving the existence of
a safe scheduler.

\begin{lemma}\label{lemma:max-scheduler}
Let $\Delta \ge \Delta'$ be such that
\[
 \discr{\Delta}\cap S = g_{\Delta',\Delta}(\discr{\Delta'} \cap S)
\]
and assume that there is a safe scheduler
in $\Sigma_{\Delta'}$. Then there is a safe scheduler in $\Sigma_{\Delta}$.
\end{lemma}
\begin{proof}
Using $g_{\Delta', \Delta}$, we can define a function $h_{\Delta',\Delta}$ from $\Sigma_{\Delta'}$ to $\Sigma_{\Delta}$
to capture our intuition of strategies that differ only on the length of a time period as follows.
Given $\sigma\in \Sigma_{\Delta'}$ and a history
\[
\seq{x_0, (m_1, i_1\cdot \Delta'), r_1,\ldots x_k}
\]
we put 
\begin{multline*}
h_{\Delta',\Delta}(\sigma)(\seq{g_{\Delta',\Delta}(x_0), (m_1, i_1\cdot \Delta), r_1,\ldots g_{\Delta',\Delta}(x_k)})\\
= \sigma(\seq{x_0, (m_1, i_1\cdot \Delta'), r_1, \ldots x_k})
\end{multline*}

Now it is easy to prove by induction that if the set of points that are reachable under $\sigma\in \Sigma_{\Delta'}$ is $X$,
then the set of points reachable by $h_{\Delta',\Delta}(\sigma)\in \Sigma_\Delta$ is equal to $g_{\Delta',\Delta}(X)$.
\end{proof}

Now we are ready to proceed with the proof of the correctness of the algorithm.
Let $\Delta_{max}$ be the actual solution, and let $\bar\Delta_{max}$ be the returned number.
We know that $\bar\Delta_{max}\le \Delta_{max}$, since the algorithm ensures that there is a safe scheduler in $\Sigma_{\bar\Delta_{max}}$.
To prove $\bar\Delta_{max} \ge \Delta_{max}$, it suffices to show that there is a safe scheduler in $\Sigma_{\Delta_{max}}$, and that $\Delta_{max}$
is found for some $\py$ at line \ref{line:lp} of Algorithm~\ref{algmax}.

To show that there is a safe scheduler in $\Sigma_{\Delta_{max}}$, let
\[
 X:= \bigcap_{\Delta < \Delta_{max}} g_{\Delta,\Delta_{max}} (\discr{\Delta}\cap S).
\]
We have $X= \discr{\Delta_{max}}\cap S$. The inclusion $\supseteq$ follows by Lemma~\ref{lemma:points-monotone};
the inclusion $\subseteq$ follows by the fact that $S$ is closed and the fact that
as $\Delta$ gets arbitrary close to $\Delta_{max}$, the points $\py\in\discr{\Delta}$ get arbitrary close to $g_{\Delta,\Delta_{max}} (\py)$.
By Lemma~\ref{lemma:points-monotone} and because $\discr{\Delta}\cap S$ is finite for all $\Gamma\le \Delta<\Delta_{max}$,
there is $\Delta<\Delta_{max}$ such that $g_{\Delta,\Delta_{max}} (\discr{\Delta}\cap S) = X$, and by definition of $\Delta_{max}$ there is a safe scheduler
in $\Sigma_\Delta$.
Finally by Lemma~\ref{lemma:max-scheduler} there must be a safe scheduler $\sigma\in \Sigma_{\Delta_{max}}$.

Now suppose that $\Delta_{max}$ is not a solution to any of the linear programs executed on line~\ref{line:lp}.
For each $\py\in \discr{\Gamma}\cap S$, let $\Delta_{\py}$ be the solution to the linear program for $\py$. Let $P$ be the set of all
$\py\in \discr{\Gamma}\cap S$ satisfying $g_{\Gamma,\Delta_{max}}(\py) \in \discr{\Delta_{max}}\cap S$. Define
$\Delta = \min_{\py\in P} \Delta_{\py}$.
We have $\Delta>\Delta_{max}$, since if $\Delta=\Delta_{max}$ then $\Delta_{max}$ would be the solution to the linear program
for the point $\py$ which realises the minimum, and if $\Delta<\Delta_{max}$ then $g_{\Gamma,\Delta_{max}}(\py) \not\in \discr{\Delta_{max}}\cap S$.
In addition,
$g_{\Delta_{max}, \Delta}(\discr{\Delta_{max}}\cap S)=\discr{\Delta}\cap S$ which by Lemma~\ref{lemma:max-scheduler} implies that there is a safe scheduler in $\Sigma_\Delta$,
contradicting the maximality of $\Delta_{max}$.

}

\end{document}